%% file: main.tex
\documentclass[sigconf, screen]{acmart}

\input{definition.tex}

\def\model{{\scshape NegGen}\xspace}
\newtcolorbox[list inside=prompt]{prompt}[1][]{
    colbacktitle=black!60,
    coltitle=white,
    fontupper=\small,
    boxsep=1pt,
    left=0pt,
    right=0pt,
    top=1pt,
    bottom=0pt,
    boxrule=1pt,
    #1,
}

\usepackage{xcolor} 

\newcommand{\ignore}[1]{}

\newcommand{\ie}{\emph{i.e., }}

\usepackage[font=small,skip=2pt]{caption}

\AtBeginDocument{%
  }

\copyrightyear{2025}
\acmYear{2025}
\setcopyright{acmlicensed}\acmConference[MM '25]{Proceedings of the 33rd ACM International Conference on Multimedia}{October 27--31, 2025}{Dublin, Ireland}
\acmBooktitle{Proceedings of the 33rd ACM International Conference on Multimedia (MM '25), October 27--31, 2025, Dublin, Ireland}
\acmDOI{10.1145/3746027.3754977}
\acmISBN{979-8-4007-2035-2/2025/10}

\begin{document}
\title{Generating Negative Samples for Multi-Modal Recommendation}

\settopmatter{authorsperrow=3}
\author{Yanbiao Ji}
  \affiliation{%
  \institution{Shanghai Jiao Tong University}
  \city{Shanghai}
  \country{China}
}
\email{jiyanbiao@sjtu.edu.cn}

\author{Dan Luo}
  \affiliation{%
  \institution{Lehigh University}
  \city{Bethlehem}
  \state{PA}
  \country{USA}
}
\email{danluo.ir@gmail.com}

\author{Chang Liu}
  \affiliation{%
  \institution{Shanghai Jiao Tong University}
  \city{Shanghai}
  \country{China}
}
\email{isonomialiu@sjtu.edu.cn}

\author{Shaokai Wu}
  \affiliation{%
  \institution{Shanghai Jiao Tong University}
  \city{Shanghai}
  \country{China}
}
\email{shaokai.wu@sjtu.edu.cn}

\author{Jing Tong}
  \affiliation{%
  \institution{Shanghai Jiao Tong University}
  \city{Shanghai}
  \country{China}
}
\email{tj_19_hf@sjtu.edu.cn}

\author{Qicheng He}
  \affiliation{%
  \institution{Shanghai Jiao Tong University}
  \city{Shanghai}
  \country{China}
}
\email{ayombeach@sjtu.edu.cn}

\author{Deyi Ji}
  \affiliation{%
  \institution{Tencent}
  \city{Beijing}
  \country{China}
}
\email{deyiji@tencent.com}

\author{Hongtao Lu}
  \affiliation{%
  \institution{Shanghai Jiao Tong University}
  \city{Shanghai}
  \country{China}
}
\email{htlu@sjtu.edu.cn}

\author{Yue Ding}
\authornote{Corresponding author.}
  \affiliation{%
  \institution{Shanghai Jiao Tong University}
  \city{Shanghai}
  \country{China}
}
\email{dingyue@sjtu.edu.cn}

\renewcommand{\shortauthors}{Yanbiao Ji et al.}

\input{Chapters/0-Abstract}

\ccsdesc[500]{Information systems~Recommender systems}

\keywords{Multi-Modal Recommendation, Negative Sampling, Large Language Models}



\maketitle

\input{Chapters/1-Introduction}
\input{Chapters/2-Preliminary}
\input{Chapters/3-Method}
\input{Chapters/4-Experiments}

\input{Chapters/5-RelatedWork}
\input{Chapters/6-Conclusion}

\begin{acks}
    This paper is supported by NSFC (No. 62176155), Shanghai Municipal Science and Technology Major Project, China, under grant No. 2021SHZDZX0102.
\end{acks}

\bibliographystyle{ACM-Reference-Format}
\balance
\bibliography{reference}


\end{document}

%% file: definition.tex
\usepackage[utf8]{inputenc}
\usepackage[T1]{fontenc}
\usepackage{microtype}

\usepackage{bm}
\usepackage{graphicx}
\usepackage{amsmath,amsfonts}
\usepackage{booktabs}
\usepackage{multirow}
\usepackage{makecell}
\usepackage{textcomp}
\usepackage{array}
\usepackage{textcomp}
\usepackage[font=footnotesize]{subfig}
\usepackage{multirow}
\usepackage{float}
\usepackage{xspace}
\usepackage[export]{adjustbox}
\usepackage{nicefrac}
\usepackage{colortbl}
\usepackage{tablefootnote}
\usepackage[referable]{threeparttablex}
\usepackage{comment}
\usepackage{enumitem}
\usepackage{balance}
\usepackage[most]{tcolorbox}
\usepackage{balance}
\usepackage{arydshln}
\setlist{nolistsep}

\usepackage{amsthm}
\newtheorem{theorem}{Theorem}
\newtheorem{lemma}[theorem]{Lemma}
\newtheorem{proposition}[theorem]{Proposition}



%% file: Chapters/0-Abstract.tex
\begin{abstract}
Multi-modal recommender systems (MMRS) have gained significant attention due to their ability to leverage information from various modalities to enhance recommendation quality. However, existing negative sampling techniques often struggle to effectively utilize the multi-modal data, leading to suboptimal performance. In this paper, we identify two key challenges in negative sampling for MMRS: (1) producing cohesive negative samples contrasting with positive samples and (2) maintaining a balanced influence across different modalities. To address these challenges, we propose \textbf{\model}, a novel framework that utilizes multi-modal large language models (MLLMs) to generate balanced and contrastive negative samples. We design three different prompt templates to enable \model to analyze and manipulate item attributes across multiple modalities, and then generate negative samples that introduce better supervision signals and ensure modality balance. Furthermore, \model employs a causal learning module to disentangle the effect of intervened key features and irrelevant item attributes, enabling fine-grained learning of user preferences. Extensive experiments on real-world datasets demonstrate the superior performance of \model compared to state-of-the-art methods in both negative sampling and multi-modal recommendation. 
\end{abstract}

%% file: Chapters/1-Introduction.tex
\vspace{-8pt}
\section{Introduction}
\label{intro}
Recommendation systems (RS) are the core component in online platforms for providing users with personalized content~\cite{recsys}.
Meanwhile, the amount of multimedia content on online platforms, including text, images, and videos, has rapidly grown~\cite{image1, text1, mmrs1, mmrs2, mmrs3}.
In this context, \textit{multi-modal recommender systems} (MMRS) have emerged with the aim of integrating multi-modal inputs with user historical behavior data to better understand user preferences and enhance recommendation quality.


Bayesian Personalized Ranking (BPR)~\cite{bpr} is a widely adopted approach for training personalized recommender models. It learns informative user and item representations that rank positive items above negative ones. Therefore, effective negative sampling strategies play an important role in optimizing these RS, which should not only accelerate convergence, but also improve model performance~\cite{dnsmn}. 
Existing negative sampling techniques can be categorized into two types~\cite{negrec}: \textit{negative item sampling} and \textit{negative item generation}. 
The former samples negatives from the item pool. For example, uniform sampling selects random uninteracted items for efficiency~\cite{bpr}.
Hard negative sampling~\cite{hard1, hard2} picks challenging candidates with large gradients to provide richer feedback and faster convergence~\cite{DBLP:journals/tois/ChenJWZFCEH21}.
The latter generates semantically meaningful and challenging negatives.
GAN-based~\cite{gan1, gan2} and diffusion-based methods~\cite{diffusion} produce hard-to-distinguish negatives, offering finer-grained supervision than item sampling.


\begin{figure}
    \centering
    \includegraphics[width=0.95\linewidth]{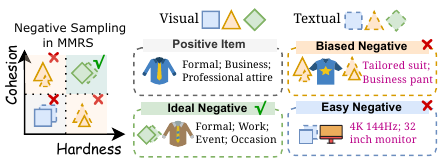}
    \caption{Illustration of effective negative sampling in MMRS. ``Cohesion'' measures the semantic contrasting to positive samples
    and ``hardness'' represents the extent of difficulty for the RS in distinguishing between positive and negative samples.
    An ideal negative sample should be both highly cohesive and sufficiently hard.}
    \label{fig:negs}
    \vspace{-15pt}
\end{figure}

Despite the effectiveness of existing negative sampling techniques, they often fall short in the context of MMRS. 
ID-based negative sampling methods~\cite{mixgcf, dens} are limited to item-level sampling strategies, and generative negative sampling methods~\cite{irgan} may encounter model collapse issues~\cite{collapse}.
We conduct an in-depth analysis of the limitations of negative sampling in MMRS and highlight our key insights as follows:
(1) \textbf{Negative samples exhibit inadequate contrast across multiple modalities}. 
We believe that the ideal negative samples in MMRS should be \textit{contrastive}, \ie exhibiting both cohesion and sufficient hardness for the RS, as illustrated in Figure~\ref{fig:negs}. 
Cohesion refers to the ability of negative samples to provide informative supervision signals that improve the performance of the model. Negative samples with high cohesion can help the model learn better representations of items~\cite{neg_proof2, diffusion}.
Hardness, on the other hand, is the similarity between negative and positive samples, \ie the difficulty of distinguishing them. Harder samples generate larger gradients, thus making the model converge faster~\cite{aobpr}.
Through an empirical study in Section~\ref{sec:pre-gen}, we demonstrate that in MMRS, if low-quality negative samples are generated in a straightforward manner without considering their cohesion and hardness, the performance can be even worse than simply using random sampling. 
(2) \textbf{Negative samples contribute unevenly to the model's learning across different modalities.} 
In Section~\ref{sec:pre-imbalance}, we investigate the learning process of a representative MMRS model FREEDOM~\cite{freedom} with negative sampling in different modalities. Our observation is that the textual modality tends to dominate the learning process over the visual modality.
This imbalance occurs because existing negative sampling strategies tend to cause multi-modal recommendation models to overfit to the easier modality (\ie text) while ignoring the more challenging ones~\cite{imbalance}. Consequently, the model fails to fully utilize the information available across all modalities, leading to a decline in performance.



Based on this, the core question to address the challenges of negative sampling for MMRS is:
\textit{How can we generate sufficiently contrastive negative samples while maintaining modality balance?}
Achieving this requires a strategy capable of adapting to diverse user preferences and comprehending complex inter-modal relationships. This naturally aligns with recent advancements in multi-modal large language models (MLLMs), which have shown  promising capabilities in understanding multi-modal inputs and generating diverse and context-aware content~\cite{table, qwen, minicpmv, mllm}. 
To this end, we propose \model, a novel framework that generates high-quality negative samples for MMRS. Unlike existing generative methods that rely on training GANs~\cite{gan} or diffusion models~\cite{diffusion}, \model leverages pretrained MLLM to create informative contrastive examples. These examples serve as challenging negative samples that enhance the learning process. 
However, pretrained MLLMs are trained on general datasets and designed for general tasks. They struggle to generate contrastive enough negative samples for recommendation tasks (Section~\ref{sec:pre-gen}). To bridge this gap, \model employs a series of tasks to generate coherent and hard negative samples: (1) Description Generation, which aggregates item attributes across modalities; (2) Attribute Masking, which identifies and masks key features of items; and (3) Attribute Completion, which replaces the masked attributes with generated alternatives. 
In these tasks we condition the generation process with the attributes of positive items across all modalities, mitigating the dominance of certain easier modalities.
Therefore \model ensures that the generated negative samples are modality-balanced and highly informative. Furthermore, \model incorporates a causal learning module to disentangle the effect of intervened key features and irrelevant item attributes, enabling a fine-grained learning of user preferences.

The main contributions of this paper are summarized as follows:
\begin{itemize}[leftmargin=*]
\item We provide an analysis of the challenges associated with negative sampling in MMRS. To the best of our knowledge, this is the first work that analyzes the limitations of existing methods under the complexities of multi-modal data.
\item To address the unique challenges of negative sampling in MMRS, we propose \model, a novel framework that leverages the capabilities of multi-modal large language models to generate semantically rich and informative negative samples.
\item We validate the effectiveness of \model through extensive experiments on multiple real-world datasets, demonstrating its superior performance over state-of-the-art methods. Notably, \model outperforms both advanced negative sampling methods and multi-modal recommender systems.
\end{itemize}

%% file: Chapters/2-Preliminary.tex
\section{Preliminary}
\label{sec:limitations}

\begin{figure}[t]
    \centering
    \includegraphics[width=0.85\linewidth]{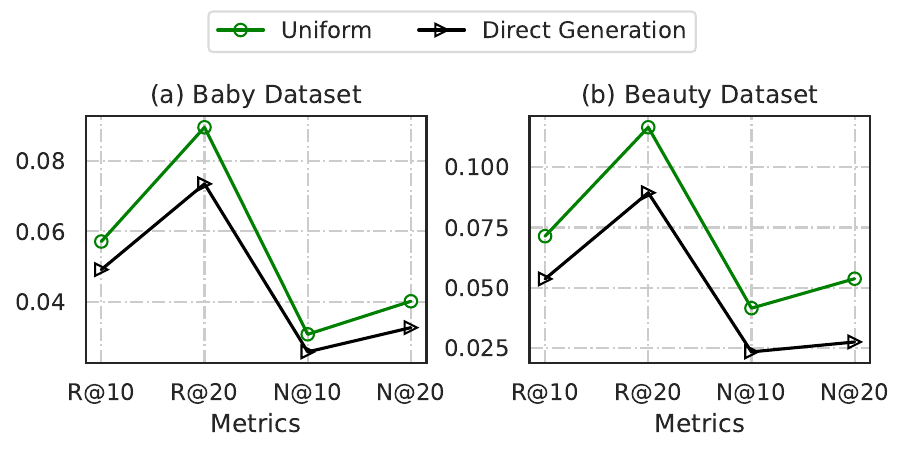}
    \caption{Recall (R) and NDCG (N) on Baby and Beauty datasets using uniform sampling and MLLM generated negative samples. }
    \vspace{-10pt}
    \label{fig:direct}
\end{figure}
\subsection{Naive Negative Sampling for MMRS}
\label{sec:pre-gen}
In this subsection we illustrate the performance of MMRS with a naive negative sampling strategy.
We choose the advanced MMRS FREEDOM~\cite{freedom} as the base recommender, and test on Amazon Baby and Beauty datasets~\footnote{\url{https://cseweb.ucsd.edu/~jmcauley/datasets.html\#amazon_reviews}}.
We use default parameter settings of FREEDOM for training. 
In particular, we leverage state-of-the-art MLLM Llama 3.2-11B-Vision\footnote{\url{https://huggingface.co/meta-llama/Llama-3.2-11B-Vision}} to generate negative samples with a structured prompt template, as presented below. 
Then, we optimize FREEDOM model using positives and generated negatives, and adopt Recall and Normalized Discounted Cumulative Gain (NDCG)~\cite{ndcg} as metrics to evaluate the performance of Top-$K$ ($K=10,20$) recommendations.
The results are illustrated in Figure~\ref{fig:direct}.  
\begin{prompt}[title=Prompt: Direct Generation]
    \textbf{Task Overview:} Given the attributes of a positive item, generate contrasting negative attributes following the same structure as the positive item. The generated attributes should be semantically similar but different from the positive item.
    
    \textbf{Input:} \textcolor{blue}{\{Positive Attributes\}}
    
    \textbf{Output:} \textcolor{blue}{\{Negative Attributes\}}
\end{prompt}

Our results reveal that training MMRS with negatives directly generated by MLLMs performs even worse than using a uniform sampling approach. 
A possible reason is that MLLMs are pre-trained on public data and general scenarios, and the negative samples generated in the specific context of recommendations lack sufficient contrast, which hinders model training.

\subsection{Imbalance of Different Modalities}

We examine the influence of different modalities on the performance of MMRS. 
We use the same model, datasets, and parameter settings as Section~\ref{sec:pre-gen}.
Specifically, we compare the performance of using a single modality (visual or textual) with that of multiple modalities.
Table~\ref{tab:pilot} presents the results. 
\label{sec:pre-imbalance}
\begin{table}[t]
    \centering
    \caption{Performance of FREEDOM with varying modalities. R and N indicate Recall and NDCG. Best results are bolded.}
    \label{tab:pilot}
    \small
    \renewcommand{\arraystretch}{0.8}
    \resizebox*{.43\textwidth}{!}{
        \begin{tabular}{@{}cccccc@{}}
            \toprule
            Datasets & Variants & R@10 & R@20 & N@10 & N@20 \\
            \midrule
            \multirow{3}{*}{Baby} & Visual\&Textual & 0.0614 & 0.0973 & 0.0306 & 0.0398 \\
            & Textual only & \textbf{0.0643} & \textbf{0.0991} & \textbf{0.0323} & {0.0411} \\
            & Visual only & 0.0574 & 0.0899 & 0.0314 & \textbf{0.0416} \\
            \midrule
            \multirow{3}{*}{Beauty} & Visual\&Textual & 0.0792& \textbf{0.1256} & 0.0455 & 0.0584 \\
            & Textual only & \textbf{0.0801} & {0.1254} & \textbf{0.0458} & \textbf{0.0597} \\
            & Visual only & 0.0763 & 0.1239 & 0.0439 & 0.0576 \\
            \bottomrule
        \end{tabular}
        }
        \vspace{-8pt}
\end{table}
An interesting finding is that using only the visual modality yields the worst performance, while using only the text modality achieves the best results, even outperforming the use of both modalities together.
To further reveal the impact of data from different modalities on model performance, we analyze it from the perspective of gradients.

The lower bound of NDCG is influenced by the embeddings of negative samples according to the precious study~\cite{ahns}:
\begin{lemma}
 Given a user $u$, the lower bound of $\text{NDCG}(u)$ can be expressed as:
 \begin{equation}
 \label{eq:bound}
 \setlength{\abovedisplayskip}{3pt}
     \text{NDCG}(u)\ge \frac{1}{|\mathcal{I}_u|}\sum_{i \in \mathcal{I}_u} \frac{1}{1+ exp(\mathbf{e}_u^T \mathbf{e}_i^* - \mathbf{e}_u^T\mathbf{e}_i)}, 
     \setlength{\belowdisplayskip}{3pt}
 \end{equation}
\end{lemma}
\noindent
where $\mathcal{I}_u$ represents the set of items interacted by user $u$, $e_u$ is the embedding of $u$, and $e_i$ and $e_i^*$ are the positive and corresponding negative item embeddings of item $i$, respectively.
We further investigate this lower bound using the gradients of negative samples under the BPR training paradigm.
\begin{proposition}
\label{bound}
    Negative samples with smaller gradient magnitudes can achieve a higher lower bound on NDCG.
\end{proposition}
\begin{proof}
    Given a user $u$ , let $\mathcal{I}_u$ be the set of items that user $u$ has interacted with, $\mathbf{e}_i$ be the embedding of item $i$, $\mathbf{e}_i^*$ be the embedding of the negative sample, and $\mathbf{e}_u$ be the embedding of user $u$. The BPR loss can be defined as:
\begin{align}
    \mathcal{L}_{\text{BPR}} &= -\sum_{i \in \mathcal{I}_u} \log \sigma(\mathbf{e}_u^T (\mathbf{e}_i - \mathbf{e}_i^*)),
\end{align}
\noindent
where $\sigma$ is the sigmoid function. The magnitude of the gradient of $\mathcal{L}_{\text{BPR}}$ with respect to $\mathbf{e}_i^*$ is:
\begin{align}
    \|\nabla\|= \left \|\frac{\partial \mathcal{L}_{\text{BPR}}}{\partial \mathbf{e}_i^*}\right \| = \|\mathbf{e}_u\|(1-\sigma(\mathbf{e}_u^T (\mathbf{e}_i - \mathbf{e}_i^*))).
\end{align}
Therefore, we have:
\begin{align}
    \label{eq:grad}
    \sigma(\mathbf{e}_u^T (\mathbf{e}_i - \mathbf{e}_i^*)) = 1 - \frac{\|\nabla\|}{\|\mathbf{e}_u\|}.
\end{align}
Combining Equation~\ref{eq:bound} and Equation~\ref{eq:grad}, we obtain:
 \begin{equation}
     \setlength{\abovedisplayskip}{3pt}
     \text{NDCG}(u)\ge \frac{1}{|\mathcal{I}_u|}\sum_{i \in \mathcal{I}_u} (1 - \frac{\|\nabla\|}{\|\mathbf{e}_u\|}). 
    \setlength{\belowdisplayskip}{3pt}
 \end{equation}
 With a smaller gradient, this lower bound becomes higher.
\end{proof}

\begin{figure}[t]
    \centering
    \includegraphics[width=0.85\linewidth]{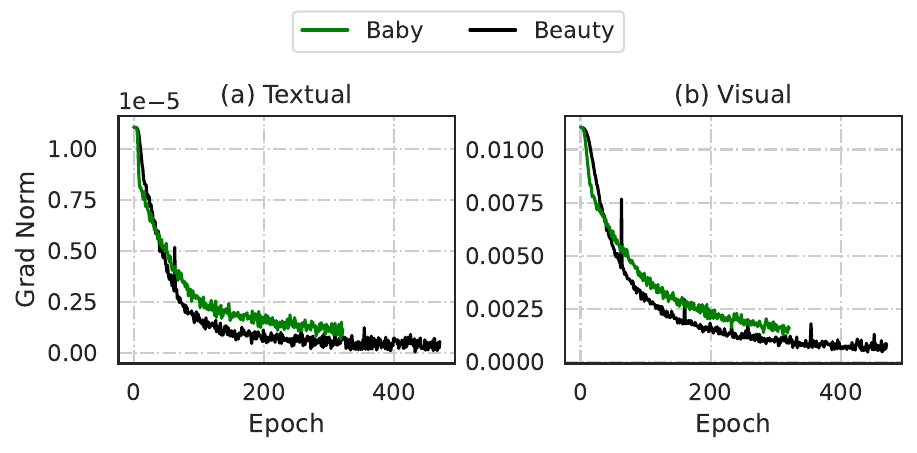}
    \caption{Gradient magnitudes of negative sample modalities over training epochs on Baby and Beauty. Textual gradients are consistently smaller than visual ones.
    }
    \vspace{-10pt}
    \label{fig:gradients}
\end{figure}


Based on Proposition~\ref{bound}, the superior performance of textual unimodal learning may stem from smaller gradient magnitudes of textual negatives. To verify this, we visualize average gradient magnitudes across modalities,  as shown in Figure~\ref{fig:gradients}. We observe that the textual modality yields smaller gradients, indicating stronger discriminative signals. This could be attributed to the rich semantics in text, which makes negative samples easier to distinguish. However, such imbalance leads the model to overly rely on certain modalities while overlooking others.



%% file: Chapters/3-Method.tex
\section{Proposed Method}
In this section, we present our framework \model. 
The overall architecture is shown in Figure~\ref{fig:framework}. 
\begin{figure*}[h]
    \centering
    \includegraphics[width=0.78\textwidth]{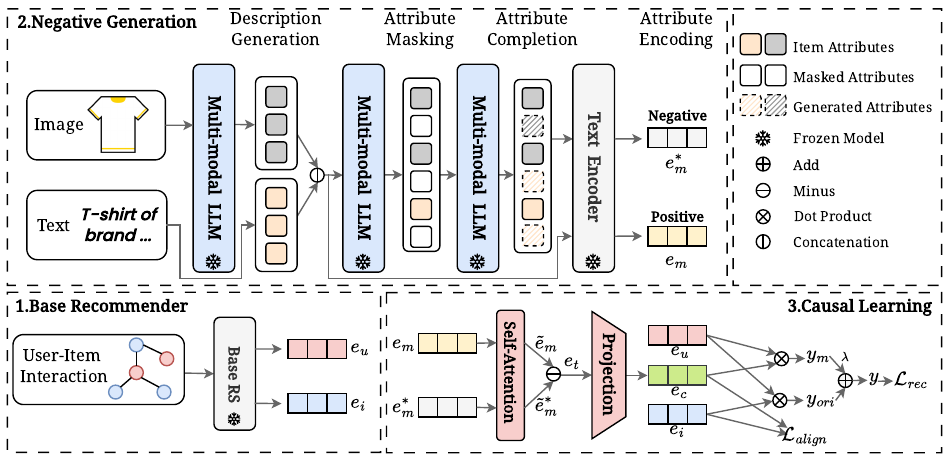}
    \caption{The overall architecture of \model comprises three main components: (1) Base Recommender Training, which captures collaborative filtering signals from user-item interactions, (2) Negative Sample Generation, which produces contrasting item attributes via MLLM, and (3) Causal Learning, which models causal relationships between multi-modal item characteristics and recommendation outcomes.}
    \label{fig:framework}
    \vspace{-8pt}
\end{figure*}

\vspace{-8pt}
\subsection{Base Recommender Training}
We first train a base recommender using user-item interactions.
The base model learns a shared embedding space for users and items to support subsequent recommendations.
Following previous works~\cite{mcln, gnn4, slmrec}, we employ LightGCN~\cite{lightgcn}, a widely applied graph-based collaborative filtering method as our base recommender. 
LightGCN leverages a simplified graph convolutional network to propagate information between connected nodes, which can be formally defined as follows:
\begin{equation}
\setlength{\abovedisplayskip}{3pt}
    \mathbf{e}_j^{(l+1)} = \sum_{k \in \mathcal{N}_j} \frac{1}{\sqrt{|\mathcal{N}_j| \cdot |\mathcal{N}_k|}} \mathbf{e}_k^{(l)},
    \setlength{\belowdisplayskip}{3pt}
\end{equation}
\noindent
where $\mathbf{e}_j^{(l+1)}$ and $\mathbf{e}_k^{(l)}$ represent the embeddings of nodes $j$ at layer $l+1$ and node $k$ at layer $l$, and $\mathcal{N}_j$ and $\mathcal{N}_k$ denote the neighbor sets of nodes $j$ and $k$, respectively. 
The final embeddings for user $u$ and item $i$ are derived by averaging the embeddings across all layers:
\begin{equation}
\setlength{\abovedisplayskip}{3pt}
    \mathbf{e}_u = \frac{1}{L+1} \sum_{l=0}^{L} \mathbf{e}_u^{(l)}, \quad \mathbf{e}_i = \frac{1}{L+1} \sum_{l=0}^{L} \mathbf{e}_i^{(l)},
    \setlength{\belowdisplayskip}{3pt}
\end{equation}
\noindent
where $L$ denotes the total number of graph convolutional layers. These aggregated embeddings effectively encode multi-hop collaborative signals for recommendation tasks.

\vspace{-8pt}
\subsection{Negative Sample Generation}
Rather than directly using MLLMs for generation, our negative sample generation module guides the process through a series of tasks to ensure the informativeness of negatives, thereby avoiding low-contrast samples as discussed in Section~\ref{sec:limitations}.
Specifically, it involves four sequential steps: Description Generation, Attribute Masking, Attribute Completion, and Attribute Encoding.


\subsubsection{Description Generation}
We extract multi-modal attributes by generating a natural language description of each item using an MLLM, which analyzes visual content to produce detailed textual descriptions. This effectively integrates underrepresented visual information into user preference learning, addressing modality imbalance. The generated visual description is then combined with existing textual metadata (\ie, brand name, product title) to form a comprehensive set of multi-modal attributes for each item.

\begin{prompt}[title={Prompt: Description Generation}, label=prompt:image]
    \textbf{Task Overview:}
    You are a descriptive writer who excels at capturing the essence and details of items in clear language. Generate natural, detailed description of the item shown in the given image.
    
    \textbf{Input:}
    \textcolor{blue}{\{Item Image\}} 

    \textbf{Output:}
    \textcolor{blue}{\{Item Description\}}
\vspace{-1pt}
    
\end{prompt}

\subsubsection{Attribute Masking}
A critical step in our method is generating templates for negative sampling. Given an item's multi-modal description, we use an MLLM to identify and mask key descriptive elements, producing partially complete attributes with [MASK] tokens while preserving the original structure and context. These masked descriptions serve as templates to guide the generation of semantic yet distinct negative samples. This ensures that negatives remain contrastive—related to but meaningfully different from their positive counterparts through contrasting attributes.

\vspace{-8pt}
\begin{prompt}[title={Prompt: Attribute Masking}, label=prompt:mask]
    \textbf{Task Overview:}
    Transform the given item description by masking key feature words with [MASK].
    
    \textbf{Instructions:}
    \begin{itemize}[leftmargin=*]
        \item Analyze the given item description.
        \item Identify most significant words that represent (1) Core features (2) Distinctive characteristics (3) Key specifications
        \item Replace these words with [MASK].
        \item Only output the masked description.
    \end{itemize}

    \textbf{Input:}
    \textcolor{blue}{\{Item Description\}}    

    \textbf{Output:}
    \textcolor{blue}{\{Masked Description\}}
\end{prompt}

\subsubsection{Attribute Completion}
Building upon the masked templates, we use MLLM to generate challenging negative samples by replacing [MASK] with appropriate alternative words. 
The negative samples generated can provide better contrastive signals for the training of recommender, 
thus enhancing its ability to capture fine-grained preferences and make more accurate recommendations.

\begin{prompt}[title={Prompt: Attribute Completion}, label=prompt:unmask]
    \textbf{Task Overview:}
    Complete the masked product description by filling in the missing words marked with [MASK].
    
    \textbf{Instructions:}
    \begin{itemize}[leftmargin=*]
        \item Analyze the given masked product description.
        \item Identify the possible words that can be used to complete the masked description.
        \item Replace the [MASK] tokens with the appropriate words.
        \item Ensure that the completed description is coherent and meaningful.
        \item Only the completed description is output.
    \end{itemize}

    \textbf{Input:}
    \textcolor{blue}{\{Masked Description\}}    

    \textbf{Output:}
    \textcolor{blue}{\{Generated Description\}}
\end{prompt}

\subsubsection{Attribute Encoding}
To capture the semantic representations of both original and generated attributes, we employ a pre-trained text encoder that maps the textual attributes into a shared vector space. The encoder, denoted as $\text{Encoder}(\cdot)$, takes the attributes of both visual and textual modalities as input and outputs dense vector representations. This process can be formally expressed as:
\begin{align}
\setlength{\belowdisplayskip}{3pt}
    \mathbf{e}_m &= \text{Encoder}(\{\text{Original Attributes}\}), \\
    \mathbf{e}_m^* &= \text{Encoder}(\{\text{Generated Attributes}\}).
\setlength{\belowdisplayskip}{3pt}
\end{align}
Here, $\mathbf{e}_m$ and $\mathbf{e}_m^*$ represent the encoded vector for the original item description and the generated negative sample, respectively.

\subsection{Causal Learning}
Though the above steps enable the generation of contrastive negative samples, the model inevitably suffers from spurious correlation, \ie the false connection between recommendation results and irrelevant attributes~\cite{spur1, spur2}.
To mitigate this spurious correlation, we propose a causal learning framework for \model. Specifically, the framework captures the causal effect of multi-modal attributes on the score prediction by estimating the total causal effect between positive and negative samples.


\subsubsection{Causal Graph Analysis}
We present the structural causal model~(SCM) for MMRS in Figure~\ref{fig:causal}, which consists of five variables: visual attributes of items ($V$), textual attributes of items ($T$), combined multi-modal representations ($M$), user representations ($U$), and the final prediction score ($Y$). The structural equations can be defined as:
\begin{align}
    M_{v,t} &= m = f_{M}(V=v, T=t), \\
    Y_{m,u} &= Y_{u,v,t} = f_{Y}(M=m, U=u) = f_{Y}(M=f_{M}(v, t), u),
\end{align}
\noindent
where uppercase letters denote random variables (\ie $M$) and lowercase letters represent their specific values (\ie $m$).
\begin{figure}[t]
    \centering
    \includegraphics[width=0.25\textwidth]{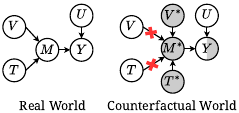}
    \caption{Causal graph illustrating the relationships between multi-modal features and recommendation outcomes.}
    \vspace{-8pt}
    \label{fig:causal}
\end{figure}

Following the do-calculus framework~\cite{scm}, we examine the causal effect of multi-modal attributes on the score prediction by measuring changes in the outcome variable $Y$ when intervening on the variables $V$ and $T$. When these variables change from $(v,t)$ to $(v^*,t^*)$, representing the transition from positive to generated negative items, the total causal effect (TE) is expressed as:
\begin{equation}
\begin{aligned}
\text{TE} &= \mathbb{E}[Y|do(V=v, T=t)] - \mathbb{E}[Y|do(V=v^*, T=t^*)] \\
&= f_{Y}(f_{M}(v, t), u) - f_{Y}(f_{M}(v^*, t^*), u).
\end{aligned}
\label{eq:te}
\end{equation}

\subsubsection{Causal Learning Module}
To instantiate the total causal effect estimation, we introduce the causal learning module, which 
implements the aforementioned causal framework through a structured comparison of positive and negative samples. To effectively capture the complex relationships between item attributes and user preferences, we introduce a self-attention mechanism that computes contextualized representations of the multi-modal embeddings:
\begin{align}
    \mathbf{Q} &= \mathbf{W}_q \mathbf{e}_m, \quad
    \mathbf{K} = \mathbf{W}_k \mathbf{e}_m, \quad
    \mathbf{V} = \mathbf{W}_v \mathbf{e}_m, \\
    \mathbf{\tilde{e}}_m &= \text{Self-Attention}(\mathbf{e}_m) = \text{softmax}(\frac{\mathbf{Q}\mathbf{K}^{\top}}{\sqrt{d}})\mathbf{V},
\end{align}
\noindent
where $d$ represents the embedding dimension and $\mathbf{W}_q$, $\mathbf{W}_k$, $\mathbf{W}_v$ are learnable parameters. The self-attention mechanism enables the model to attend to different aspects of the item's features, capturing their relative importance and reduce the impact of irrelevant attributes~\cite{attn}. For consistency, we apply identical transformations to the negative sample features:
\begin{equation}
    \mathbf{\tilde{e}}_m^* = \text{Self-Attention}(\mathbf{e}_m^*).
\end{equation}
Then we instantiate the total causal effect in Equation~\ref{eq:te}:
\begin{equation}
    \mathbf{e}_{t} = \mathbf{\tilde{e}}_m - \mathbf{\tilde{e}}_m^*.
\end{equation}
\noindent
This formulation serves two key purposes: (1) it captures fine-grained preference distinctions by explicitly modeling the contrast between positive and negative samples, and (2) it reduces the impact of spurious factors by focusing on causal relationships between important item attributes and user preferences~\cite{causal}.

To align the causal effect embedding $\mathbf{e}_{t}$ with the dimension of the base recommender embeddings, we project it into the same space using a projection network:
\begin{equation}
    \mathbf{e}_c = ReLU(\mathbf{e}_{t} \mathbf{W}_1 + \mathbf{b}_1) \mathbf{W}_2 + \mathbf{b}_2,
\end{equation}
\noindent
where $\mathbf{W}_1$ and $\mathbf{W}_2$ are learnable weight matrices, and $\mathbf{b}_1$ and $\mathbf{b}_2$ are learnable biases. Similarly, the negative embedding $\mathbf{\tilde{e}}_m^*$ is projected into the same space:
\begin{equation}
    \mathbf{{e}}_c^* = ReLU(\mathbf{\tilde{e}}_{m}^* \mathbf{W}_1 + \mathbf{b}_1) \mathbf{W}_2 + \mathbf{b}_2.
\end{equation}

\subsection{Training Objective}
To optimize \model, we define two complementary objectives: a recommendation loss $\mathcal{L}_{rec}$ to ensure the recommendation quality, and a multi-modal alignment loss $\mathcal{L}_{align}$ to enhance the consistency between item features and their multi-modal causal representations. This subsection elaborates on the details of these objectives.

First, \model combines item scores predicted from collaborative embeddings and multi-modal causal representation. The causal score is computed as the inner product of the user embedding $e_u$ and the causal effect embedding $e_c$:
$
    y_m = \mathbf{e}_u^{\top} \mathbf{e}_c.
$
\noindent
In parallel, the base recommendation score, which captures the collaborative filtering signal, is calculated as:
    $y_{ori} = \mathbf{e}_u^{\top} \mathbf{e}_i,$
where $\mathbf{e}_i$ is the item embedding learned by the pre-trained recommender.
The final prediction score is obtained by linearly combining the causal effect score and the base recommendation score:
$
    y = y_{ori} + \lambda y_m.
$
Here, $\lambda$ is a hyper-parameter that controls the contribution of the causal effect score relative to the base score.

Having obtained the score prediction, we can calculate the recommendation loss $\mathcal{L}_{rec}$. We adopt the Bayesian Personalized Ranking (BPR) loss explicitly modeling a ranking objective:
\begin{equation}
    \mathcal{L}_{rec} = \mathbb{E}_{(u, i) \sim \mathcal{D}} \left[ -\log \sigma(   \lambda \mathbf{e}_u^{\top} \mathbf{e}_c + \mathbf{e}_u^{\top} \mathbf{e}_i - \mathbf{e}_u^{\top} \mathbf{e}_c^*) \right],
\end{equation}
\noindent
where $\mathcal{D}$ denotes the set of training samples $(u, i)$, with $e_c$ as the positive embedding of item $i$ and $e_c^*$ as the corresponding negative embedding. The term $\sigma(\cdot)$ represents the sigmoid function. This loss encourages the model to rank positive items higher than negative ones for each user $u$.

In addition to the ranking loss, we introduce a multi-modal alignment loss to ensure that the causal embeddings derived from multi-modal features are well-aligned with the collaborative embeddings. This is achieved using a contrastive learning objective:
\begin{equation}
    \mathcal{L}_{align} = \mathbb{E}_{i \in \mathcal{D}} \left[- \log \frac{\exp(\mathbf{e}_c^{\top} \mathbf{e}_i / \tau)}{ \exp((\mathbf{e}_c^*)^{\top} \mathbf{e}_i / \tau)} \right],
\end{equation}
where $\mathcal{D}$ represents the set of items, $\mathbf{e}_i$ is the collaborative embedding of item $i$, $\mathbf{e}_c$ and $\mathbf{e}_c^*$ are the corresponding positive and negative multi-modal embedding, respectively, and $\tau$ is a temperature hyper-parameter that controls the sharpness of the logits distribution~\cite{contrastive}. This loss encourages closer alignment between multi-modal embeddings and their corresponding item embeddings while pushing apart negative samples.

The overall objective function integrates both losses, allowing the model to balance collaborative filtering signals with multi-modal alignment effects:
\begin{equation}
    \mathcal{L} = \mathcal{L}_{rec} + \alpha \mathcal{L}_{align}.
\end{equation}
Here, $\alpha$ is a hyper-parameter that adjusts the relative importance of the alignment objective. By optimizing this joint objective, the model effectively leverages both collaborative and multi-modal information to improve recommendation performance.

%% file: Chapters/4-Experiments.tex
\begin{table*}[ht]
    \centering
    \caption{Overall performance achieved by different recommendation methods in terms of Recall and NDCG. The best performance is highlighted in bold and the second best is \underline{underlined}. $\Delta Improv.$ indicates relative improvements over the second best method. \textit{p-val} denotes the p-value from paired t-tests comparing \model against the best baseline.}
    \label{tab:overall}
    \renewcommand{\arraystretch}{0.83}
    \resizebox*{\textwidth}{!}{
\begin{tabular}{@{}ccccccccccccccccc@{}}
\toprule
\multirow{2.5}{*}{Model}                 & \multicolumn{4}{c}{Baby} & \multicolumn{4}{c}{Beauty} & \multicolumn{4}{c}{Clothing}  & \multicolumn{4}{c}{Sports}     \\ 
\cmidrule(l){2-5} \cmidrule(l){6-9} \cmidrule(l){10-13} \cmidrule(l){14-17} 
                       & R@10 & R@20 & N@10 & N@20 & R@10 & R@20 & N@10 & N@20 & R@10 & R@20 & N@10 & N@20 & R@10 & R@20 & N@10 & N@20  \\ \cmidrule(r){1-17}
BPR & 0.0359 & 0.0572 & 0.0193 & 0.0244 & 0.0568 & 0.0953 & 0.0274 & 0.0311 & 0.0207 & 0.0302 & 0.0114 & 0.0137 &  0.0431 & 0.0650 & 0.0238 & 0.0297  \\                      
LightGCN & 0.0458 & 0.0698 & 0.0236 & 0.0309 &0.0676 & 0.1041 & 0.0403 & 0.0519  & 0.0328 & 0.0504 & 0.0184 & 0.0226 & 0.0547 & 0.0819 & 0.0305 & 0.0381\\
\hdashline
IRGAN&0.0441&0.0673&0.0268&0.0322&0.0569&0.0951&0.0379&0.0484&0.0312&0.0481&0.0186&0.0245&0.0477&0.0715&0.0322&0.0390\\
MixGCF&0.0456&0.0691&0.0281&0.0341&0.0617&0.0978&0.0406&0.0557&0.0321&0.0492&0.0199&0.0277&0.0529&0.0773&0.0339&0.0435\\
DENS&0.0451&0.0708&0.0287&0.0350&0.0629&0.0988&0.0415&0.0568&0.0326&0.0501&0.0204&0.0282&0.0537&0.0784&0.0343&0.0446\\
DNS(M,N)&0.0457&0.0725&0.0290&0.0358&0.0628&0.0986&0.0417&0.0566&0.0328&0.0510&0.0211&0.0286&0.0540&0.0785&0.0337&0.0431\\
AHNS&0.0465&0.0737&0.0292&0.0361&0.0635&0.0992&0.0426&0.0572&0.0335&0.0528&0.0237&0.0291&0.0554&0.0851&0.0352&0.0462\\
 \hdashline
VBPR & 0.0412 & 0.0637 & 0.0203 & 0.0251 & 0.0637 & 0.1002 & 0.0371 & 0.0471 & 0.0283 & 0.0411 & 0.0156 & 0.0189 & 0.0558 & 0.0853 & 0.0298 & 0.0375\\
MMGCN & 0.0441 & 0.0685 & 0.0192 & 0.0264 & 0.0628 & 0.1005 & 0.0341 & 0.0459 & 0.0227 & 0.0360 & 0.0121 & 0.0154 & 0.0380 & 0.0631 & 0.0201 & 0.0272\\
BM3 & 0.0542 & 0.0873 & 0.0296 & 0.0365 & 0.0728 & 0.1187 & 0.0425 & 0.0567 & 0.0451 & 0.0672 & 0.0242 & 0.0297 & 0.0663 & 0.0981 & 0.0354 & 0.0438\\
FREEDOM & 0.0614 & 0.0973 & 0.0306 & 0.0398 & 0.0771 & 0.1247 & 0.0441 & 0.0579 & 0.0617 & 0.0904 & 0.0317 & 0.0409 & 0.0714 & 0.1073 & 0.0368 & 0.0457\\
DRAGON & \underline{0.0643} & 0.0996 & 0.0325 & 0.0412 & \underline{0.0792} & \underline{0.1256} & \underline{0.0455} & \underline{0.0584} & 0.0621 & \underline{0.0915} & \underline{0.0322} & \underline{0.0418} & \underline{0.0732} & \underline{0.1085} & \underline{0.0389} & \underline{0.0482}\\
DiffMM & 0.0625 & 0.0971 & 0.0306 & 0.0401 & 0.0783 & 0.1249 & 0.0447 & 0.0579 & \underline{0.0622} & 0.0907 & 0.0313 & 0.0397 & 0.0715 & 0.1043 & 0.0350 & 0.0431\\
LGMRec & 0.0631 & \underline{0.1002} & \underline{0.0327} & \underline{0.0422} & 0.0781 & 0.1243 & 0.0448 & 0.0578 & 0.0531 & 0.0817 & 0.0296 & 0.0364 & 0.0723 & 0.1066 & 0.0385 & 0.0477\\
\hdashline
\model & \textbf{0.0701} & \textbf{0.1065} & \textbf{0.0342} & \textbf{0.0438} & \textbf{0.0823} & \textbf{0.1285} & \textbf{0.0473} & \textbf{0.0604} & \textbf{0.0654} & \textbf{0.0961} & \textbf{0.0350} & \textbf{0.0441} & \textbf{0.0763} & \textbf{0.1114} & \textbf{0.0411} & \textbf{0.0506} \\
\rowcolor{gray!20}  $\Delta Improv.$ & 9.02\% & 6.29\% & 7.95\% & 4.59\% & 3.91\% & 2.31\% & 3.96\% & 3.42\% & 5.14\% & 5.03\% & 8.70\% & 5.50\% & 4.23\% & 3.96\% & 5.66\% & 4.98\%\\ 
\textit{p-val} & 0.0093 & 0.0080 & 0.0028 & 0.0263 & 0.0221 & 0.0101 & 0.0017 & 0.0003 & 0.0098 & 0.0068 & 0.0423 & 0.0062 & 0.0044 & 0.0026 & 0.0079 & 0.0045\\ 
\bottomrule
\end{tabular}}
\end{table*}
\section{Experiments}
In this section, we conduct extensive experiments to demonstrate the effectiveness of \model. 
In general, we expect
the experimental results to answer the following research questions:
\textbf{RQ1}: How does \model perform compared with state-of-the-art MMRS methods and negative sampling methods?
\textbf{RQ2}: Can \model effectively utilize item information from different modalities?
\textbf{RQ3}: Can \model produce high-quality negative samples that accelerate convergence and improve the performance of the recommender?
\textbf{RQ4}: How do the individual components of our model contribute to its performance across different datasets?
\textbf{RQ5}: How does different choices of hyper-parameters affect the performance of our model?

\subsection{Experimental Settings}

\subsubsection{Datasets}
We evaluate our approach on four distinct categories from the Amazon review dataset~\cite{amazon}: Baby, Beauty, Clothing, and Sports. Each dataset represents a different domain of consumer behavior and purchasing patterns. To ensure data quality and meaningful user-item interactions, we use the 5-core filtered datasets, retaining only users and items with at least 5 interactions, following the practices in~\cite{vbpr, freedom, dragon}. Each item in the datasets is associated with rich multi-modal information, including product images and metadata (\ie title, category and brand). The statistics of the datasets are shown in Table~\ref{tab:dataset}.

\begin{table}[h]
    \centering
    \vspace{-8pt}
    \caption{Statistics of the datasets.}
    \label{tab:dataset}
    \renewcommand{\arraystretch}{0.8}
    \begin{tabular}{ccccc}
        \toprule
        Dataset & \#Users & \#Items & \#Interactions & Density \\
        \midrule
        Baby & 19,445 & 7,050 & 139,110 & 0.00101 \\
        Beauty & 22,363 & 12,101 & 172,188 & 0.00064 \\
        Clothing & 39,387 & 23,033 & 278,677 & 0.00031 \\
        Sports & 33,598 & 18,357 & 296,337 & 0.00048 \\
        \bottomrule
    \end{tabular}
    \vspace{-8pt}
\end{table}

\subsubsection{Evaluation Protocols}
Following previous works~\cite{eval1, fettle}, we adopt the 80-10-10 split protocol for training, validation, and testing. Two widely used Top-$K$ metrics, \ie Recall~(R@$K$) and NDCG~(N@$K$), are used to evaluate the quality of recommendation. We report the average values of all users in the test set with $K=10, 20$.

\subsubsection{Baseline Models}
We compare our method with the following three groups of baseline methods.
\noindent \textit{(1) Collaborative Filtering Baselines.}
\textbf{BPR}~\cite{bpr} learns user preferences by ranking interacted items higher than uninteracted ones. \textbf{LightGCN}~\cite{lightgcn} simplifies graph convolution networks by removing non-linear activation and feature transformation.
\noindent \textit{(2) Negative Sampling Baselines}.
\textbf{IRGAN}~\cite{irgan} uses a minimax game to optimize generative and discriminative networks. \textbf{MixGCF}~\cite{mixgcf} synthesizes hard negatives with neighborhood embeddings. \textbf{DENS}~\cite{dens} uses factor-aware sampling to identify the best negatives. \textbf{DNS(M,N)}~\cite{dnsmn} controls negative sampling hardness with hyperparameters. \textbf{AHNS}~\cite{ahns} selects negatives with adaptive hardnesses during training.
\noindent(3) Multi-Modal Recommender Systems.
\textbf{VBPR}~\cite{vbpr} incorporates visual information into BPR embeddings. \textbf{MMGCN}~\cite{gnn4} performs message passing in each modality. \textbf{BM3}~\cite{bm3} creates contrastive views via dropout. \textbf{FREEDOM}~\cite{freedom} denoises user-item graph during training. \textbf{DRAGON}~\cite{dragon} learns dual representations on heterogeneous and homogeneous graphs. \textbf{DiffMM}~\cite{diffmm} integrates a modality-aware graph diffusion model. \textbf{LGMRec}~\cite{lgmrec} proposes hypergraph embedding module.


\subsubsection{Implementation Details}
For all models, we set the embedding dimension to 64 for both users and items, following~\cite{dragon, lgmrec}. We initialize the model parameters using the Xavier method~\cite{xavier} and optimize using Adam~\cite{adam} with fixed batch size of 2048. The training process runs for a maximum of 1000 epochs, with early stopping applied after 20 epochs without improvement. We use Recall@20 on the validation set as our stopping criterion, consistent with~\cite{dragon, freedom}. In \model, we use Llama 3.2-11B-Vision model as the MLLM. For a fair comparison, we use the Sentence-Bert~\cite{sbert} model to encode textual attributes as used by Amazon datasets. LightGCN is used for \model as well as all baselines requiring a base recommender. We implement all models in PyTorch~\cite{pytorch} and conduct experiments on one NVIDIA H100 GPU with 80GB memory.
 
\vspace{-7pt}
\subsection{Overall Performance Comparison~(RQ1)}
The experimental results are summarized in Table~\ref{tab:overall}, and we have the following observations.
(1) \textbf{\model significantly outperforms all MMRS and negative sampling methods across evaluation metrics}.
The performance gains are particularly significant in the Baby dataset, where we observe a 9.02\% improvement in Recall@10, and in the Clothing dataset, with an 8.70\% enhancement in NDCG@10. Even in scenarios where baseline methods achieve strong performance, such as in the Beauty and Sports datasets, \model remains competitive with improvements ranging from 2.31\% to 5.66\% across various metrics. These comprehensive improvements show the effectiveness of \model in producing high quality negative samples for MMRS and learning multi-modal user preferences. 
(2) \textbf{\model exhibits strong generalization across diverse domains.} 
\model achieves consistent improvements across heterogeneous datasets, highlighting its adaptability and effectiveness. The method shows remarkable effectiveness in both the Sports product category, with relatively less modality dependence~\cite{fettle}, and the Clothing category, which requires rich multimedia content such as images to describe complicated fashion designs. This dual success in handling both weak and strong modality dependency scenarios establishes \model as a versatile solution for diverse multi-modal recommendation applications.
(3) \textbf{Traditional collaborative filtering methods are insufficient to adapt to multi-modal scenarios}. 
Experimental results show that specialized MMRS frameworks, such as DRAGON and LGMRec, significantly outperform traditional models like BPR and LightGCN, highlighting the importance of effective multi-modal integration. In contrast, simple extensions of conventional models (\ie, VBPR and MMGCN) offer only marginal gains, underscoring the complexity of multi-modal data and the need for dedicated negative sampling strategies rather than direct adaptations of existing methods.

\begin{table}[t]
    \centering
    \caption{
    Effect of different modalities on \model's leaning. R and N denote Recall and NDCG, respectively. The best performance is in bold. ``V'' and ``T'' represent visual and textual inputs, respectively.
    }
    \label{tab:unimodal}
    \renewcommand{\arraystretch}{0.8}
    \resizebox*{.4\textwidth}{!}{
    \small
        \begin{tabular}{@{}cccccc@{}}
            \toprule
            Datasets & Variants & R@10 & R@20 & N@10 & N@20 \\
            \midrule
            \multirow{3}{*}{Baby} & V & 0.0667 & 0.0934 & 0.0307 & 0.0401 \\
            & T & 0.0678 & 0.0956 & 0.0327 & 0.0412 \\
            \rowcolor{gray!20} \cellcolor{white} & V\&T & \textbf{0.0701} & \textbf{0.1065} & \textbf{0.0342} & \textbf{0.0438} \\
            \midrule
            \multirow{3}{*}{Beauty} & V & 0.0794 & 0.1213 & 0.0426 & 0.0548 \\
            & T & 0.0803 & 0.1231 & 0.0439 & 0.0586 \\
            \rowcolor{gray!20} \cellcolor{white} &V\&T & \textbf{0.0832} & \textbf{0.1285} & \textbf{0.0473} & \textbf{0.0604} \\
            \midrule
            \multirow{3}{*}{Clothing} & V & 0.0621 & 0.0910 & 0.0327 & 0.0412 \\
            & T & 0.0632 & 0.0930 & 0.0338 & 0.0426 \\
            \rowcolor{gray!20} \cellcolor{white} & V\&T  &\textbf{0.0654} & \textbf{0.0961} & \textbf{0.0350} & \textbf{0.0441} \\
            \midrule
            \multirow{3}{*}{Sports} & V & 0.0667 & 0.0940 & 0.0336 & 0.0415 \\
            & T & 0.0687 & 0.1064 & 0.0357 & 0.0443 \\
            \rowcolor{gray!20} \cellcolor{white} & V\&T & \textbf{0.0763} & \textbf{0.1114} & \textbf{0.0411} & \textbf{0.0506} \\
            \bottomrule
            
        \end{tabular}
        }
        \vspace{-5pt}
\end{table}

\vspace{-5pt}
\subsection{Modality Ablation Study~(RQ2)}
According to our preliminary experiments in Section~\ref{sec:limitations}, existing MMRS can easily suffer from imbalanced learning of different modalities. This observation aligns with the findings from \citet{imbalance} that the dominance of certain modalities can lead to underfitting of other modalities during training. 
To demonstrate that \model effectively mitigates this issue, we conduct a series of experiments, which compares the performance of \model, denoted by V\&T with its unimodal counterpart, denoted by V and T. The results are shown in Table~\ref{tab:unimodal}.
The V\&T variant, which combines visual and textual information, achieves the best performance compared to the other two variants, \ie V and T. 
This outcome suggests that \model can effectively leverage the complementary attributes from both modalities to generate accurate recommendations that align with users' multi-modal preferences. 
And by introducing fine-grained attribute-level negative samples, \model demonstrates its capability to address the modality imbalance.

\begin{figure}[t]
    \centering
    \includegraphics[width=0.85\linewidth]{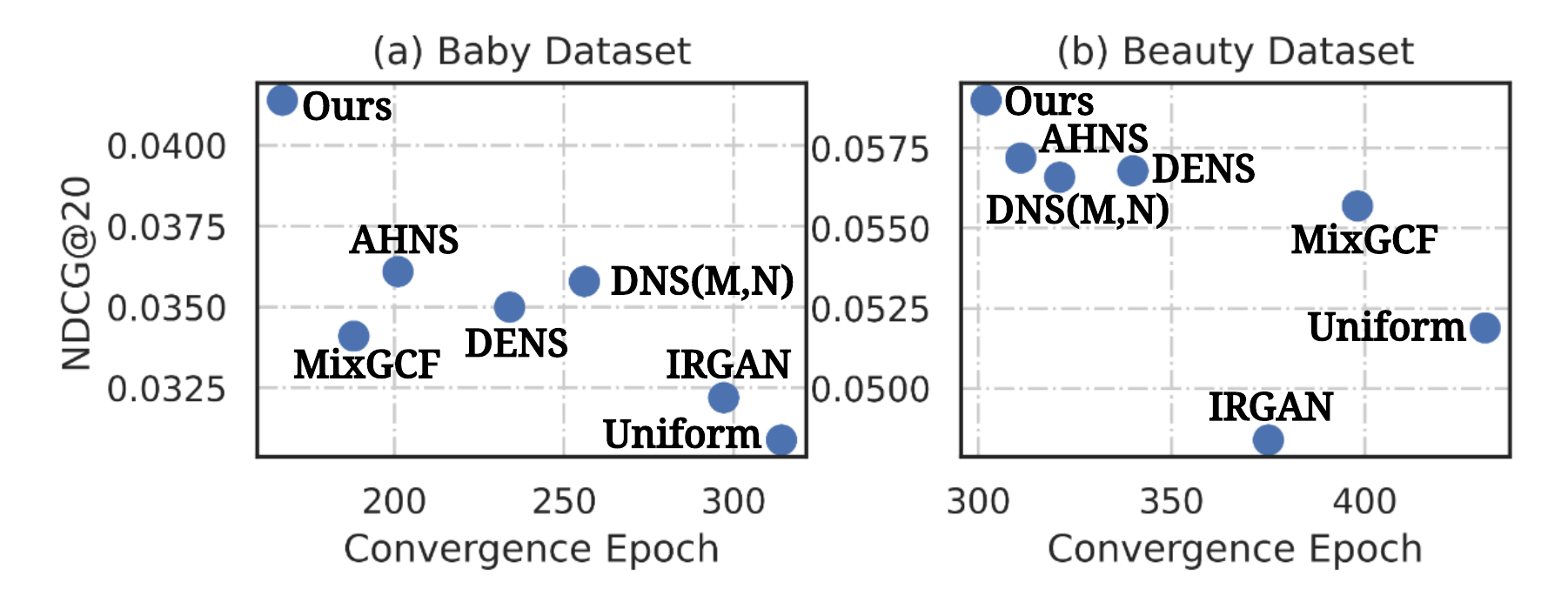}
    \caption{Visualization of the convergence epoch and NDCG@20 for different negative sampling methods. Higher NDCG and lower convergence epochs indicate higher-quality negative samples.}
    \label{fig:quality}
    \vspace{-8pt}
\end{figure}

\subsection{Negative Sample Quality~(RQ3)}
To evaluate the quality of the negative samples generated by \model in terms of cohesion and hardness, we compare its performance (\ie NDCG@20) against its convergence epoch number in Figure~\ref{fig:quality} to different negative sampling methods. For all methods we use the LightGCN model as base recommender. In particular, we remove the the causal learning module in \model for a fair comparison. The results show that \model effectively balances cohesion and hardness in negative samples. Its high NDCG@20 indicates that the negatives are cohesive, providing meaningful contrast, while the faster convergence suggests that the negatives are sufficiently hard, allowing the model to learn discriminative features quickly.

\begin{table}[h]
    \centering
    \vspace{-3pt}
    \caption{Ablation study of different components in \model. R and N denote Recall and NDCG, respectively, with the best performance highlighted in bold.}
    \label{tab:ablation}
    \renewcommand{\arraystretch}{0.8}
    \resizebox*{.4\textwidth}{!}{
    \small
        \begin{tabular}{@{}cccccc@{}}
            \toprule
            Datasets & Variants & R@10 & R@20 & N@10 & N@20 \\
            \midrule
            \multirow{4}{*}{Baby} & w/o NEG & 0.0571 & 0.0896 & 0.0307 & 0.0401 \\
            & w/o CF & 0.0641 & 0.0992 & 0.0321 & 0.0414 \\
            & w/o both & 0.0466 & 0.0684 & 0.0235 & 0.0317 \\
            \rowcolor{gray!20} \cellcolor{white} & \model & \textbf{0.0701} & \textbf{0.1065} & \textbf{0.0342} & \textbf{0.0438} \\
            \midrule
            \multirow{4}{*}{Beauty} & w/o NEG & 0.0714 & 0.1165 & 0.0416 & 0.0538 \\
            & w/o CF & 0.0805 & 0.1263 & 0.0457 & 0.0590 \\
            & w/o both & 0.0682 & 0.1083 & 0.0402 & 0.0531 \\
            \rowcolor{gray!20} \cellcolor{white} & \model & \textbf{0.0832} & \textbf{0.1285} & \textbf{0.0473} & \textbf{0.0604} \\
            \midrule
            \multirow{4}{*}{Clothing} & w/o NEG & 0.0626 & 0.0910 & 0.0327 & 0.0420 \\
            & w/o CF & 0.0637 & 0.0922 & 0.0341 & 0.0432 \\
            & w/o both & 0.0336 & 0.0531 & 0.0187 & 0.0232 \\
            \rowcolor{gray!20} \cellcolor{white} & \model & \textbf{0.0654} & \textbf{0.0961} & \textbf{0.0350} & \textbf{0.0441} \\
            \midrule
            \multirow{4}{*}{Sports} & w/o NEG & 0.0627 & 0.0940 & 0.0336 & 0.0415 \\
            & w/o CF & 0.0714 & 0.1055 & 0.0374 & 0.0471 \\
            & w/o both & 0.0539 & 0.0813 & 0.0304 & 0.0375 \\
            \rowcolor{gray!20} \cellcolor{white} & \model & \textbf{0.0763} & \textbf{0.1114} & \textbf{0.0411} & \textbf{0.0506} \\
            \bottomrule
        \end{tabular}
        }
    \vspace{-4pt}
\end{table}


\subsection{Module Ablation Study~(RQ4)}
To evaluate the effectiveness of each component in \model, we conduct ablation studies by removing the negative sample generation mechanism (\textit{w/o NEG}), the causal learning framework (\textit{w/o CF}), and both components simultaneously (\textit{w/o both}). 
For the \textit{w/o NEG} variant, we replace the negative sample generation with the uniform negative sampling as used in~\cite{vbpr, freedom}. 
For the \textit{w/o CF} variant, we remove the causal learning framework and directly project the obtained attribute embeddings. Table~\ref{tab:ablation} summarizes the results, with the following observations.
(1) \textbf{Both modules are effective in improving multi-modal recommendation.} Removing either the negative sample generation module (\textit{w/o NEG}) or the causal learning framework (\textit{w/o CF}) causes a significant drop in performance across all metrics, demonstrating their effectiveness in capturing nuanced user preferences and learning robust representations.
(2) \textbf{The two components of \model are complementary.} When both components are removed, the performance decline is even more pronounced, confirming that negative sample generation and causal learning work synergistically to enhance \model's ability to model user preferences accurately. 




\subsection{Hyper-parameter Analysis (RQ5)}

In this subsection, we conduct experiments to examine the impact of key hyper-parameters in \model. Specifically, we analyze the effect of (1) the weight of the multi-modal prediction score $\lambda$, (2) the weight of the multi-modal alignment loss $\alpha$, and (3) the temperature $\tau$ in the alignment loss. The results are presented in Figure~\ref{fig:hyper}. 
\paragraph{Effect of $\lambda$} 
The parameter $\lambda$ controls the impact of the multi-modal prediction score. Across all datasets, performance improves as $\lambda$ increases from 0 to 0.4, but stabilizes or slightly declines when $\lambda > 0.6$, indicating that the optimal range for $\lambda$ is between 0.4 and 0.6. Notably, setting $\lambda$ to 0 causes a significant performance drop, highlighting the importance of incorporating multi-modal signals to enhance recommendation quality.

\paragraph{Effect of $\alpha$} 
The parameter $\alpha$ controls the weight of the multi-modal alignment loss. Our experiments show that \model performs stably and robustly across a wide range of $\alpha$ values, from $1\text{e-4}$ to 1. However, performance slightly degrades when $\alpha$ becomes too large, indicating that overemphasizing the alignment loss may disrupt the balance among training objectives.

\paragraph{Effect of $\tau$} 
The temperature parameter $\tau$ in the alignment loss modulates the sharpness of the logits distribution. Varying $\tau$ across [0.07, 0.1, 0.2, 0.3, 1], we find that smaller values generally lead to better performance. Specifically, when $\tau$ is set to 1, the logits distribution becomes overly smooth, hindering the model's ability to learn meaningful alignments. 
\begin{figure}[t]
    \centering
    \includegraphics[width=0.43\textwidth]{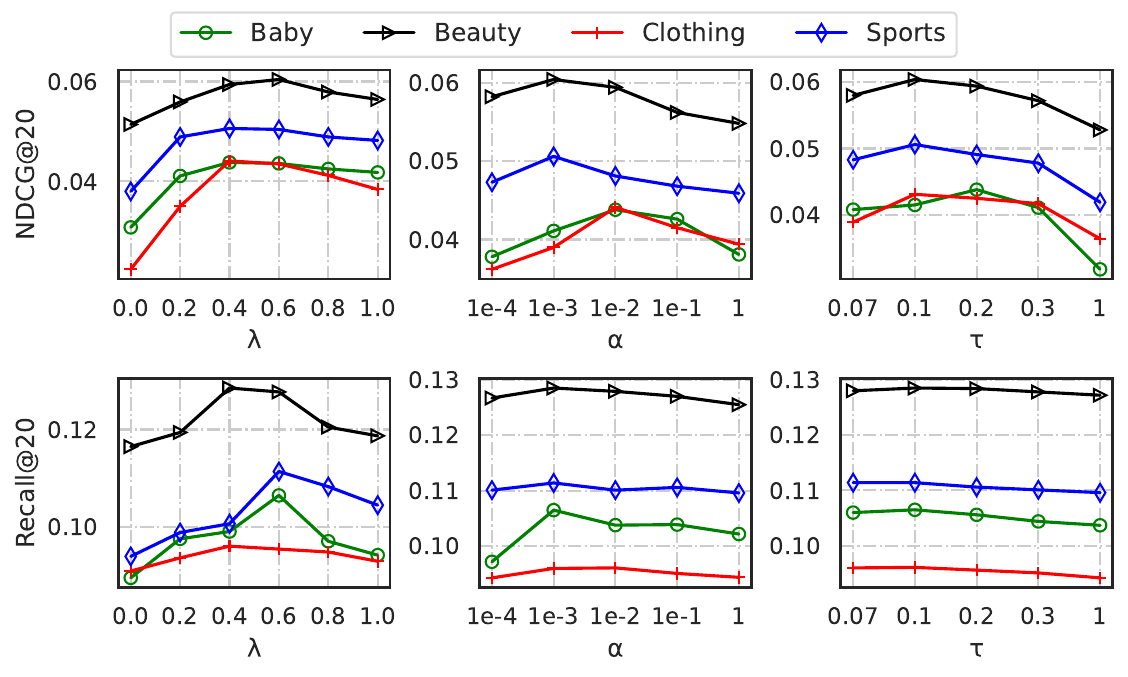}
    \vspace{-6pt}
    \caption{Effect of different hyper-parameters on \model.}
    \vspace{-6pt}
    \label{fig:hyper}
\end{figure}

%% file: Chapters/5-RelatedWork.tex
\section{Related Work}
\paragraph{Hard Negative Sampling in RS}
Negative sampling is widely used across fields like computer vision~\cite{cv1, cv2, cv3}, natural language processing~\cite{nlp1, nlp2, nlp3}, and information retrieval~\cite{ir1, ir2, ir3}. In RS, hard negative sampling (HNS) techniques are used to improve negative sample quality. Early HNS methods~\cite{pinsage, dns, aobpr} focused on ranking uninteracted items to identify suitable negatives, such as the dynamic selection mechanism in DNS~\cite{dns}.
HNS approaches are now divided into sampling-based and generation-based methods~\cite{survey}. Sampling-based approaches utilize advanced algorithms like reinforcement learning~\cite{rl1, rl2} and social network analysis~\cite{sn1, sn2}. And generation-based methods, such as GANs~\cite{gan, gan1, gan2}, synthesize challenging negatives. Applying these methods to MMRS remains difficult due to the complexity of multi-modal tasks.
\vspace{-5pt}
\paragraph{MMRS}
The growth of multimedia data has driven MMRS. Early research~\cite{mcf1, mcf2, mcf3} extended collaborative filtering by incorporating multi-modal content as side information. For example, VBPR~\cite{vbpr} integrates visual features into the BPR~\cite{bpr} framework.
More recently, Graph Neural Networks (GNNs) have become powerful tools for MMRS~\cite{gnn1, gnn2, gnn3}, enabling the modeling of complex relationships between users, items, and modalities. MMGCN~\cite{gnn4}, for instance, builds modality-specific user-item graphs and enriches node embeddings via neighbor aggregation. While negative sampling remains common, the effectiveness of sampling strategies in multi-modal contexts remains underexplored.
\vspace{-5pt}
\paragraph{Causal Inference for RS}
Recently, there has been growing interest in incorporating causal reasoning into RS~\cite{causalrs}. Several works~\cite{causal1, causal2, causal3} have used causal frameworks to address recommendation biases, such as popularity bias~\cite{pop1, pop2} and exposure bias~\cite{expo1, expo2}. For instance, MACR~\cite{macr} reduces popularity bias by modeling cause-effect relationships in recommendation using a causal graph. Causal inference also helps to model the underlying causality behind user preferences and behavior, mitigate spurious correlations~\cite{invrl, mcln} and enable explainable recommendation~\cite{explain1, explain2}.

%% file: Chapters/6-Conclusion.tex
\section{Conclusion and Future Work}
In this paper, we first identify the challenge of producing both cohesive and hard negative samples in multi-modal recommender systems. Then we address it by proposing \model, which uses the strong ability of MLLMs to generate balanced, contrastive negative samples. Experiments show \model significantly improves recommendation performance across metrics and datasets. 

%% file: main.bbl

\begin{thebibliography}{86}


\ifx \showCODEN    \undefined \def \showCODEN     #1{\unskip}     \fi
\ifx \showDOI      \undefined \def \showDOI       #1{#1}\fi
\ifx \showISBNx    \undefined \def \showISBNx     #1{\unskip}     \fi
\ifx \showISBNxiii \undefined \def \showISBNxiii  #1{\unskip}     \fi
\ifx \showISSN     \undefined \def \showISSN      #1{\unskip}     \fi
\ifx \showLCCN     \undefined \def \showLCCN      #1{\unskip}     \fi
\ifx \shownote     \undefined \def \shownote      #1{#1}          \fi
\ifx \showarticletitle \undefined \def \showarticletitle #1{#1}   \fi
\ifx \showURL      \undefined \def \showURL       {\relax}        \fi
\providecommand\bibfield[2]{#2}
\providecommand\bibinfo[2]{#2}
\providecommand\natexlab[1]{#1}
\providecommand\showeprint[2][]{arXiv:#2}

\bibitem[Bai et~al\mbox{.}(2024)]%
        {mmrs3}
\bibfield{author}{\bibinfo{person}{Haoyue Bai}, \bibinfo{person}{Le Wu}, \bibinfo{person}{Min Hou}, \bibinfo{person}{Miaomiao Cai}, \bibinfo{person}{Zhuangzhuang He}, \bibinfo{person}{Yuyang Zhou}, \bibinfo{person}{Richang Hong}, {and} \bibinfo{person}{Meng Wang}.} \bibinfo{year}{2024}\natexlab{}.
\newblock \showarticletitle{Multimodality Invariant Learning for Multimedia-Based New Item Recommendation} \emph{(\bibinfo{series}{SIGIR '24})}. \bibinfo{pages}{677–686}.
\newblock


\bibitem[Bai et~al\mbox{.}(2023)]%
        {qwen}
\bibfield{author}{\bibinfo{person}{Jinze Bai}, \bibinfo{person}{Shuai Bai}, \bibinfo{person}{Shusheng Yang}, \bibinfo{person}{Shijie Wang}, \bibinfo{person}{Sinan Tan}, \bibinfo{person}{Peng Wang}, \bibinfo{person}{Junyang Lin}, \bibinfo{person}{Chang Zhou}, {and} \bibinfo{person}{Jingren Zhou}.} \bibinfo{year}{2023}\natexlab{}.
\newblock \bibinfo{title}{Qwen-VL: A Versatile Vision-Language Model for Understanding, Localization, Text Reading, and Beyond}.
\newblock
\showeprint[arxiv]{2308.12966}~[cs.CV]


\bibitem[Chen et~al\mbox{.}(2021)]%
        {DBLP:journals/tois/ChenJWZFCEH21}
\bibfield{author}{\bibinfo{person}{Jiawei Chen}, \bibinfo{person}{Chengquan Jiang}, \bibinfo{person}{Can Wang}, \bibinfo{person}{Sheng Zhou}, \bibinfo{person}{Yan Feng}, \bibinfo{person}{Chun Chen}, \bibinfo{person}{Martin Ester}, {and} \bibinfo{person}{Xiangnan He}.} \bibinfo{year}{2021}\natexlab{}.
\newblock \showarticletitle{CoSam: An Efficient Collaborative Adaptive Sampler for Recommendation}.
\newblock \bibinfo{journal}{\emph{{ACM} Trans. Inf. Syst.}} \bibinfo{volume}{39}, \bibinfo{number}{3} (\bibinfo{year}{2021}), \bibinfo{pages}{34:1--34:24}.
\newblock


\bibitem[Chen et~al\mbox{.}(2019)]%
        {sn1}
\bibfield{author}{\bibinfo{person}{Jiawei Chen}, \bibinfo{person}{Can Wang}, \bibinfo{person}{Sheng Zhou}, \bibinfo{person}{Qihao Shi}, \bibinfo{person}{Yan Feng}, {and} \bibinfo{person}{Chun Chen}.} \bibinfo{year}{2019}\natexlab{}.
\newblock \showarticletitle{SamWalker: Social Recommendation with Informative Sampling Strategy} \emph{(\bibinfo{series}{WWW '19})}. \bibinfo{pages}{228–239}.
\newblock


\bibitem[Chen et~al\mbox{.}(2017)]%
        {mcf1}
\bibfield{author}{\bibinfo{person}{Jingyuan Chen}, \bibinfo{person}{Hanwang Zhang}, \bibinfo{person}{Xiangnan He}, \bibinfo{person}{Liqiang Nie}, \bibinfo{person}{Wei Liu}, {and} \bibinfo{person}{Tat-Seng Chua}.} \bibinfo{year}{2017}\natexlab{}.
\newblock \showarticletitle{Attentive Collaborative Filtering: Multimedia Recommendation with Item- and Component-Level Attention} \emph{(\bibinfo{series}{SIGIR '17})}. \bibinfo{pages}{335–344}.
\newblock


\bibitem[Czapp et~al\mbox{.}(2024)]%
        {mmrs1}
\bibfield{author}{\bibinfo{person}{\'{A}d\'{a}m~Tibor Czapp}, \bibinfo{person}{M\'{a}ty\'{a}s Jani}, \bibinfo{person}{B\'{a}lint Domi\'{a}n}, {and} \bibinfo{person}{Bal\'{a}zs Hidasi}.} \bibinfo{year}{2024}\natexlab{}.
\newblock \showarticletitle{Dynamic Product Image Generation and Recommendation at Scale for Personalized E-commerce} \emph{(\bibinfo{series}{RecSys '24})}. \bibinfo{pages}{768–770}.
\newblock


\bibitem[Deng et~al\mbox{.}(2024)]%
        {causal1}
\bibfield{author}{\bibinfo{person}{Jianfeng Deng}, \bibinfo{person}{Qingfeng Chen}, \bibinfo{person}{Debo Cheng}, \bibinfo{person}{Jiuyong Li}, \bibinfo{person}{Lin Liu}, {and} \bibinfo{person}{Xiaojing Du}.} \bibinfo{year}{2024}\natexlab{}.
\newblock \showarticletitle{Mitigating Dual Latent Confounding Biases in Recommender Systems}.
\newblock
\showeprint[arxiv]{2410.12451}~[cs.IR]


\bibitem[Devasier et~al\mbox{.}(2024)]%
        {nlp1}
\bibfield{author}{\bibinfo{person}{Jacob Devasier}, \bibinfo{person}{Yogesh Gurjar}, {and} \bibinfo{person}{Chengkai Li}.} \bibinfo{year}{2024}\natexlab{}.
\newblock \showarticletitle{Robust Frame-Semantic Models with Lexical Unit Trees and Negative Samples} \emph{(\bibinfo{series}{ACL '24})}.
\newblock


\bibitem[Ding et~al\mbox{.}(2019)]%
        {rl2}
\bibfield{author}{\bibinfo{person}{Jingtao Ding}, \bibinfo{person}{Yuhan Quan}, \bibinfo{person}{Xiangnan He}, \bibinfo{person}{Yong Li}, {and} \bibinfo{person}{Depeng Jin}.} \bibinfo{year}{2019}\natexlab{}.
\newblock \showarticletitle{Reinforced negative sampling for recommendation with exposure data} \emph{(\bibinfo{series}{IJCAI'19})}. \bibinfo{pages}{99–109}.
\newblock


\bibitem[Du et~al\mbox{.}(2022)]%
        {invrl}
\bibfield{author}{\bibinfo{person}{Xiaoyu Du}, \bibinfo{person}{Zike Wu}, \bibinfo{person}{Fuli Feng}, \bibinfo{person}{Xiangnan He}, {and} \bibinfo{person}{Jinhui Tang}.} \bibinfo{year}{2022}\natexlab{}.
\newblock \showarticletitle{Invariant Representation Learning for Multimedia Recommendation} \emph{(\bibinfo{series}{MM '22})}. \bibinfo{pages}{619–628}.
\newblock


\bibitem[Fan et~al\mbox{.}(2023)]%
        {hard1}
\bibfield{author}{\bibinfo{person}{Lu Fan}, \bibinfo{person}{Jiashu Pu}, \bibinfo{person}{Rongsheng Zhang}, {and} \bibinfo{person}{Xiao{-}Ming Wu}.} \bibinfo{year}{2023}\natexlab{}.
\newblock \showarticletitle{Neighborhood-based Hard Negative Mining for Sequential Recommendation} \emph{(\bibinfo{series}{SIGIR'23})}. \bibinfo{pages}{2042–2046}.
\newblock


\bibitem[Gao et~al\mbox{.}(2024)]%
        {causal}
\bibfield{author}{\bibinfo{person}{Chen Gao}, \bibinfo{person}{Yu Zheng}, \bibinfo{person}{Wenjie Wang}, \bibinfo{person}{Fuli Feng}, \bibinfo{person}{Xiangnan He}, {and} \bibinfo{person}{Yong Li}.} \bibinfo{year}{2024}\natexlab{}.
\newblock \showarticletitle{Causal Inference in Recommender Systems: A Survey and Future Directions}.
\newblock \bibinfo{journal}{\emph{ACM Trans. Inf. Syst.}} \bibinfo{volume}{42}, \bibinfo{number}{4}, Article \bibinfo{articleno}{88} (\bibinfo{date}{Feb.} \bibinfo{year}{2024}), \bibinfo{numpages}{32}~pages.
\newblock
\showISSN{1046-8188}


\bibitem[Glorot and Bengio(2010)]%
        {xavier}
\bibfield{author}{\bibinfo{person}{Xavier Glorot} {and} \bibinfo{person}{Yoshua Bengio}.} \bibinfo{year}{2010}\natexlab{}.
\newblock \showarticletitle{Understanding the difficulty of training deep feedforward neural networks}. In \bibinfo{booktitle}{\emph{Proceedings of the Thirteenth International Conference on Artificial Intelligence and Statistics}} \emph{(\bibinfo{series}{Proceedings of Machine Learning Research})}. \bibinfo{publisher}{PMLR}.
\newblock


\bibitem[Goodfellow et~al\mbox{.}(2014)]%
        {gan}
\bibfield{author}{\bibinfo{person}{Ian~J. Goodfellow}, \bibinfo{person}{Jean Pouget-Abadie}, \bibinfo{person}{Mehdi Mirza}, \bibinfo{person}{Bing Xu}, \bibinfo{person}{David Warde-Farley}, \bibinfo{person}{Sherjil Ozair}, \bibinfo{person}{Aaron Courville}, {and} \bibinfo{person}{Yoshua Bengio}.} \bibinfo{year}{2014}\natexlab{}.
\newblock \bibinfo{title}{Generative Adversarial Networks}.
\newblock
\showeprint[arxiv]{1406.2661}~[stat.ML]


\bibitem[Gui et~al\mbox{.}(2023)]%
        {collapse}
\bibfield{author}{\bibinfo{person}{Jie Gui}, \bibinfo{person}{Zhenan Sun}, \bibinfo{person}{Yonggang Wen}, \bibinfo{person}{Dacheng Tao}, {and} \bibinfo{person}{Jieping Ye}.} \bibinfo{year}{2023}\natexlab{}.
\newblock \showarticletitle{A Review on Generative Adversarial Networks: Algorithms, Theory, and Applications}.
\newblock \bibinfo{journal}{\emph{IEEE Trans. on Knowl. and Data Eng.}} \bibinfo{volume}{35}, \bibinfo{number}{4} (\bibinfo{date}{April} \bibinfo{year}{2023}), \bibinfo{pages}{3313–3332}.
\newblock


\bibitem[Guo et~al\mbox{.}(2024)]%
        {lgmrec}
\bibfield{author}{\bibinfo{person}{Zhiqiang Guo}, \bibinfo{person}{Jianjun Li}, \bibinfo{person}{Guohui Li}, \bibinfo{person}{Chaoyang Wang}, \bibinfo{person}{Si Shi}, {and} \bibinfo{person}{Bin Ruan}.} \bibinfo{year}{2024}\natexlab{}.
\newblock \showarticletitle{LGMRec: Local and Global Graph Learning for Multimodal Recommendation} \emph{(\bibinfo{series}{AAAI '24})}. \bibinfo{pages}{8454--8462}.
\newblock


\bibitem[He and McAuley(2016)]%
        {vbpr}
\bibfield{author}{\bibinfo{person}{Ruining He} {and} \bibinfo{person}{Julian McAuley}.} \bibinfo{year}{2016}\natexlab{}.
\newblock \showarticletitle{VBPR: visual Bayesian Personalized Ranking from implicit feedback} \emph{(\bibinfo{series}{AAAI'16})}. \bibinfo{pages}{144–150}.
\newblock


\bibitem[He et~al\mbox{.}(2020)]%
        {lightgcn}
\bibfield{author}{\bibinfo{person}{Xiangnan He}, \bibinfo{person}{Kuan Deng}, \bibinfo{person}{Xiang Wang}, \bibinfo{person}{Yan Li}, \bibinfo{person}{YongDong Zhang}, {and} \bibinfo{person}{Meng Wang}.} \bibinfo{year}{2020}\natexlab{}.
\newblock \showarticletitle{LightGCN: Simplifying and Powering Graph Convolution Network for Recommendation} \emph{(\bibinfo{series}{SIGIR '20})}. \bibinfo{pages}{639–648}.
\newblock


\bibitem[He et~al\mbox{.}(2022)]%
        {spur2}
\bibfield{author}{\bibinfo{person}{Xiangnan He}, \bibinfo{person}{Yang Zhang}, \bibinfo{person}{Fuli Feng}, \bibinfo{person}{Chonggang Song}, \bibinfo{person}{Lingling Yi}, \bibinfo{person}{Guohui Ling}, {and} \bibinfo{person}{Yongdong Zhang}.} \bibinfo{year}{2022}\natexlab{}.
\newblock \showarticletitle{Addressing Confounding Feature Issue for Causal Recommendation}.
\newblock \bibinfo{journal}{\emph{ACM Trans. Inf. Syst.}} (\bibinfo{year}{2022}).
\newblock
\showISSN{1046-8188}


\bibitem[Huang et~al\mbox{.}(2021)]%
        {mixgcf}
\bibfield{author}{\bibinfo{person}{Tinglin Huang}, \bibinfo{person}{Yuxiao Dong}, \bibinfo{person}{Ming Ding}, \bibinfo{person}{Zhen Yang}, \bibinfo{person}{Wenzheng Feng}, \bibinfo{person}{Xinyu Wang}, {and} \bibinfo{person}{Jie Tang}.} \bibinfo{year}{2021}\natexlab{}.
\newblock \showarticletitle{MixGCF: An Improved Training Method for Graph Neural Network-based Recommender Systems} \emph{(\bibinfo{series}{KDD '21})}.
\newblock


\bibitem[J\"{a}rvelin and Kek\"{a}l\"{a}inen(2002)]%
        {ndcg}
\bibfield{author}{\bibinfo{person}{Kalervo J\"{a}rvelin} {and} \bibinfo{person}{Jaana Kek\"{a}l\"{a}inen}.} \bibinfo{year}{2002}\natexlab{}.
\newblock \showarticletitle{Cumulated gain-based evaluation of IR techniques}.
\newblock \bibinfo{journal}{\emph{ACM Trans. Inf. Syst.}} \bibinfo{volume}{20}, \bibinfo{number}{4} (\bibinfo{date}{Oct.} \bibinfo{year}{2002}), \bibinfo{pages}{422–446}.
\newblock
\showISSN{1046-8188}


\bibitem[Ji et~al\mbox{.}(2025)]%
        {text1}
\bibfield{author}{\bibinfo{person}{Deyi Ji}, \bibinfo{person}{Feng Zhao}, \bibinfo{person}{Hongtao Lu}, \bibinfo{person}{Feng Wu}, {and} \bibinfo{person}{Jieping Ye}.} \bibinfo{year}{2025}\natexlab{}.
\newblock \showarticletitle{Structural and Statistical Texture Knowledge Distillation and Learning for Segmentation}.
\newblock \bibinfo{journal}{\emph{IEEE Trans. Pattern Anal. Mach. Intell.}} \bibinfo{volume}{47}, \bibinfo{number}{5} (\bibinfo{date}{Jan.} \bibinfo{year}{2025}), \bibinfo{pages}{3639–3656}.
\newblock
\showISSN{0162-8828}
\urldef\tempurl%
\url{https://doi.org/10.1109/TPAMI.2025.3536481}
\showDOI{\tempurl}


\bibitem[Ji et~al\mbox{.}(2024b)]%
        {image1}
\bibfield{author}{\bibinfo{person}{Deyi Ji}, \bibinfo{person}{Feng Zhao}, \bibinfo{person}{Lanyun Zhu}, \bibinfo{person}{Wenwei Jin}, \bibinfo{person}{Hongtao Lu}, {and} \bibinfo{person}{Jieping Ye}.} \bibinfo{year}{2024}\natexlab{b}.
\newblock \bibinfo{title}{Discrete Latent Perspective Learning for Segmentation and Detection}.
\newblock
\showeprint[arxiv]{2406.10475}~[cs.CV]
\urldef\tempurl%
\url{https://arxiv.org/abs/2406.10475}
\showURL{%
\tempurl}


\bibitem[Ji et~al\mbox{.}(2024c)]%
        {table}
\bibfield{author}{\bibinfo{person}{Deyi Ji}, \bibinfo{person}{Lanyun Zhu}, \bibinfo{person}{Siqi Gao}, \bibinfo{person}{Peng Xu}, \bibinfo{person}{Hongtao Lu}, \bibinfo{person}{Jieping Ye}, {and} \bibinfo{person}{Feng Zhao}.} \bibinfo{year}{2024}\natexlab{c}.
\newblock \bibinfo{title}{Tree-of-Table: Unleashing the Power of LLMs for Enhanced Large-Scale Table Understanding}.
\newblock
\showeprint[arxiv]{2411.08516}~[cs.CL]
\urldef\tempurl%
\url{https://arxiv.org/abs/2411.08516}
\showURL{%
\tempurl}


\bibitem[Ji et~al\mbox{.}(2024a)]%
        {pop2}
\bibfield{author}{\bibinfo{person}{Yanbiao Ji}, \bibinfo{person}{Yue Ding}, \bibinfo{person}{Chang Liu}, \bibinfo{person}{Yuxiang Lu}, \bibinfo{person}{Xin Xin}, {and} \bibinfo{person}{Hongtao Lu}.} \bibinfo{year}{2024}\natexlab{a}.
\newblock \showarticletitle{Topology-Aware Popularity Debiasing via Simplicial Complexes}.
\newblock
\showeprint[arxiv]{2411.13892}~[cs.IR]


\bibitem[Jiang et~al\mbox{.}(2024)]%
        {diffmm}
\bibfield{author}{\bibinfo{person}{Yangqin Jiang}, \bibinfo{person}{Lianghao Xia}, \bibinfo{person}{Wei Wei}, \bibinfo{person}{Da Luo}, \bibinfo{person}{Kangyi Lin}, {and} \bibinfo{person}{Chao Huang}.} \bibinfo{year}{2024}\natexlab{}.
\newblock \showarticletitle{DiffMM: Multi-Modal Diffusion Model for Recommendation} \emph{(\bibinfo{series}{MM '24})}. \bibinfo{pages}{7591–7599}.
\newblock


\bibitem[Jin et~al\mbox{.}(2020)]%
        {gan2}
\bibfield{author}{\bibinfo{person}{Binbin Jin}, \bibinfo{person}{Defu Lian}, \bibinfo{person}{Zheng Liu}, \bibinfo{person}{Qi Liu}, \bibinfo{person}{Jianhui Ma}, \bibinfo{person}{Xing Xie}, {and} \bibinfo{person}{Enhong Chen}.} \bibinfo{year}{2020}\natexlab{}.
\newblock \showarticletitle{Sampling-decomposable generative adversarial recommender} \emph{(\bibinfo{series}{NIPS '20})}. Article \bibinfo{articleno}{1897}, \bibinfo{numpages}{11}~pages.
\newblock


\bibitem[Kalantidis et~al\mbox{.}(2020)]%
        {cv2}
\bibfield{author}{\bibinfo{person}{Yannis Kalantidis}, \bibinfo{person}{Mert~Bulent Sariyildiz}, \bibinfo{person}{Noe Pion}, \bibinfo{person}{Philippe Weinzaepfel}, {and} \bibinfo{person}{Diane Larlus}.} \bibinfo{year}{2020}\natexlab{}.
\newblock \showarticletitle{Hard Negative Mixing for Contrastive Learning} \emph{(\bibinfo{series}{NeurIPS '20})}. Article \bibinfo{articleno}{1829}, \bibinfo{numpages}{12}~pages.
\newblock


\bibitem[Kingma and Ba(2015)]%
        {adam}
\bibfield{author}{\bibinfo{person}{Diederik~P. Kingma} {and} \bibinfo{person}{Jimmy Ba}.} \bibinfo{year}{2015}\natexlab{}.
\newblock \showarticletitle{Adam: A Method for Stochastic Optimization} \emph{(\bibinfo{series}{ICLR '15})}.
\newblock


\bibitem[Krause et~al\mbox{.}(2024)]%
        {expo1}
\bibfield{author}{\bibinfo{person}{Thorsten Krause}, \bibinfo{person}{Alina Deriyeva}, \bibinfo{person}{Jan~H. Beinke}, \bibinfo{person}{Gerrit~Y. Bartels}, {and} \bibinfo{person}{Oliver Thomas}.} \bibinfo{year}{2024}\natexlab{}.
\newblock \showarticletitle{Mitigating Exposure Bias in Recommender Systems—A Comparative Analysis of Discrete Choice Models}.
\newblock \bibinfo{journal}{\emph{ACM Trans. Recomm. Syst.}} \bibinfo{volume}{3}, \bibinfo{number}{2}, Article \bibinfo{articleno}{19} (\bibinfo{date}{Nov.} \bibinfo{year}{2024}), \bibinfo{numpages}{37}~pages.
\newblock


\bibitem[Lai et~al\mbox{.}(2023)]%
        {dens}
\bibfield{author}{\bibinfo{person}{Riwei Lai}, \bibinfo{person}{Li Chen}, \bibinfo{person}{Yuhan Zhao}, \bibinfo{person}{Rui Chen}, {and} \bibinfo{person}{Qilong Han}.} \bibinfo{year}{2023}\natexlab{}.
\newblock \showarticletitle{Disentangled Negative Sampling for Collaborative Filtering} \emph{(\bibinfo{series}{WSDM '23})}. \bibinfo{pages}{96–104}.
\newblock


\bibitem[Lai et~al\mbox{.}(2025)]%
        {ahns}
\bibfield{author}{\bibinfo{person}{Riwei Lai}, \bibinfo{person}{Rui Chen}, \bibinfo{person}{Qilong Han}, \bibinfo{person}{Chi Zhang}, {and} \bibinfo{person}{Li Chen}.} \bibinfo{year}{2025}\natexlab{}.
\newblock \showarticletitle{Adaptive hardness negative sampling for collaborative filtering} \emph{(\bibinfo{series}{AAAI'24})}. Article \bibinfo{articleno}{961}, \bibinfo{numpages}{8}~pages.
\newblock


\bibitem[Li et~al\mbox{.}(2024a)]%
        {attn}
\bibfield{author}{\bibinfo{person}{Hongkang Li}, \bibinfo{person}{Meng Wang}, \bibinfo{person}{Tengfei Ma}, \bibinfo{person}{Sijia Liu}, \bibinfo{person}{ZAIXI ZHANG}, {and} \bibinfo{person}{Pin-Yu Chen}.} \bibinfo{year}{2024}\natexlab{a}.
\newblock \showarticletitle{What Improves the Generalization of Graph Transformers? A Theoretical Dive into the Self-attention and Positional Encoding} \emph{(\bibinfo{series}{ICML'24})}.
\newblock


\bibitem[Li et~al\mbox{.}(2023)]%
        {mcln}
\bibfield{author}{\bibinfo{person}{Shuaiyang Li}, \bibinfo{person}{Dan Guo}, \bibinfo{person}{Kang Liu}, \bibinfo{person}{Richang Hong}, {and} \bibinfo{person}{Feng Xue}.} \bibinfo{year}{2023}\natexlab{}.
\newblock \showarticletitle{Multimodal Counterfactual Learning Network for Multimedia-based Recommendation} \emph{(\bibinfo{series}{SIGIR '23})}. \bibinfo{pages}{1539–1548}.
\newblock


\bibitem[Li et~al\mbox{.}(2020)]%
        {gnn1}
\bibfield{author}{\bibinfo{person}{Xingchen Li}, \bibinfo{person}{Xiang Wang}, \bibinfo{person}{Xiangnan He}, \bibinfo{person}{Long Chen}, \bibinfo{person}{Jun Xiao}, {and} \bibinfo{person}{Tat-Seng Chua}.} \bibinfo{year}{2020}\natexlab{}.
\newblock \bibinfo{title}{Hierarchical Fashion Graph Network for Personalized Outfit Recommendation}.
\newblock
\showeprint[arxiv]{2005.12566}~[cs.IR]


\bibitem[Li et~al\mbox{.}(2024b)]%
        {fettle}
\bibfield{author}{\bibinfo{person}{Yang Li}, \bibinfo{person}{Qi'Ao Zhao}, \bibinfo{person}{Chen Lin}, \bibinfo{person}{Jinsong Su}, {and} \bibinfo{person}{Zhilin Zhang}.} \bibinfo{year}{2024}\natexlab{b}.
\newblock \showarticletitle{Who To Align With: Feedback-Oriented Multi-Modal Alignment in Recommendation Systems} \emph{(\bibinfo{series}{SIGIR '24})}. \bibinfo{pages}{667–676}.
\newblock


\bibitem[Liu et~al\mbox{.}(2017)]%
        {mcf3}
\bibfield{author}{\bibinfo{person}{Qiang Liu}, \bibinfo{person}{Shu Wu}, {and} \bibinfo{person}{Liang Wang}.} \bibinfo{year}{2017}\natexlab{}.
\newblock \showarticletitle{DeepStyle: Learning User Preferences for Visual Recommendation} \emph{(\bibinfo{series}{SIGIR '17})}. \bibinfo{pages}{841–844}.
\newblock


\bibitem[Liu et~al\mbox{.}(2021)]%
        {gnn2}
\bibfield{author}{\bibinfo{person}{Yiyu Liu}, \bibinfo{person}{Qian Liu}, \bibinfo{person}{Yu Tian}, \bibinfo{person}{Changping Wang}, \bibinfo{person}{Yanan Niu}, \bibinfo{person}{Yang Song}, {and} \bibinfo{person}{Chenliang Li}.} \bibinfo{year}{2021}\natexlab{}.
\newblock \showarticletitle{Concept-Aware Denoising Graph Neural Network for Micro-Video Recommendation} \emph{(\bibinfo{series}{CIKM '21})}. \bibinfo{pages}{1099–1108}.
\newblock


\bibitem[Ma et~al\mbox{.}(2023)]%
        {hard2}
\bibfield{author}{\bibinfo{person}{Haokai Ma}, \bibinfo{person}{Ruobing Xie}, \bibinfo{person}{Lei Meng}, \bibinfo{person}{Xin Chen}, \bibinfo{person}{Xu Zhang}, \bibinfo{person}{Leyu Lin}, {and} \bibinfo{person}{Jie Zhou}.} \bibinfo{year}{2023}\natexlab{}.
\newblock \showarticletitle{Exploring False Hard Negative Sample in Cross-Domain Recommendation} \emph{(\bibinfo{series}{RecSys'23})}. \bibinfo{pages}{502–514}.
\newblock


\bibitem[Ma et~al\mbox{.}(2024a)]%
        {negrec}
\bibfield{author}{\bibinfo{person}{Haokai Ma}, \bibinfo{person}{Ruobing Xie}, \bibinfo{person}{Lei Meng}, \bibinfo{person}{Fuli Feng}, \bibinfo{person}{Xiaoyu Du}, \bibinfo{person}{Xingwu Sun}, \bibinfo{person}{Zhanhui Kang}, {and} \bibinfo{person}{Xiangxu Meng}.} \bibinfo{year}{2024}\natexlab{a}.
\newblock \bibinfo{title}{Negative Sampling in Recommendation: A Survey and Future Directions}.
\newblock
\showeprint[arxiv]{2409.07237}~[cs.IR]


\bibitem[Ma et~al\mbox{.}(2024b)]%
        {mmrs2}
\bibfield{author}{\bibinfo{person}{Haokai Ma}, \bibinfo{person}{Yimeng Yang}, \bibinfo{person}{Lei Meng}, \bibinfo{person}{Ruobing Xie}, {and} \bibinfo{person}{Xiangxu Meng}.} \bibinfo{year}{2024}\natexlab{b}.
\newblock \showarticletitle{Multimodal Conditioned Diffusion Model for Recommendation} \emph{(\bibinfo{series}{WWW '24})}. \bibinfo{pages}{1733–1740}.
\newblock


\bibitem[Mansoury et~al\mbox{.}(2024)]%
        {expo2}
\bibfield{author}{\bibinfo{person}{Masoud Mansoury}, \bibinfo{person}{Bamshad Mobasher}, {and} \bibinfo{person}{Herke van Hoof}.} \bibinfo{year}{2024}\natexlab{}.
\newblock \showarticletitle{Mitigating Exposure Bias in Online Learning to Rank Recommendation: A Novel Reward Model for Cascading Bandits} \emph{(\bibinfo{series}{CIKM'24})}. \bibinfo{pages}{1638–1648}.
\newblock


\bibitem[McAuley et~al\mbox{.}(2015)]%
        {amazon}
\bibfield{author}{\bibinfo{person}{Julian McAuley}, \bibinfo{person}{Christopher Targett}, \bibinfo{person}{Qinfeng Shi}, {and} \bibinfo{person}{Anton van~den Hengel}.} \bibinfo{year}{2015}\natexlab{}.
\newblock \showarticletitle{Image-Based Recommendations on Styles and Substitutes} \emph{(\bibinfo{series}{SIGIR '15})}. \bibinfo{pages}{43–52}.
\newblock


\bibitem[McKechnie et~al\mbox{.}(2024)]%
        {ir1}
\bibfield{author}{\bibinfo{person}{Jack McKechnie}, \bibinfo{person}{Graham McDonald}, {and} \bibinfo{person}{Craig Macdonald}.} \bibinfo{year}{2024}\natexlab{}.
\newblock \showarticletitle{Bi-Objective Negative Sampling for Sensitivity-Aware Search} \emph{(\bibinfo{series}{SIGIR '24})}. \bibinfo{pages}{2296–2300}.
\newblock


\bibitem[Nguyen and Fang(2024)]%
        {diffusion}
\bibfield{author}{\bibinfo{person}{Trung-Kien Nguyen} {and} \bibinfo{person}{Yuan Fang}.} \bibinfo{year}{2024}\natexlab{}.
\newblock \showarticletitle{Diffusion-based Negative Sampling on Graphs for Link Prediction} \emph{(\bibinfo{series}{WWW '24})}. \bibinfo{pages}{948–958}.
\newblock


\bibitem[Paszke et~al\mbox{.}(2019)]%
        {pytorch}
\bibfield{author}{\bibinfo{person}{Adam Paszke}, \bibinfo{person}{Sam Gross}, \bibinfo{person}{Francisco Massa}, \bibinfo{person}{Adam Lerer}, \bibinfo{person}{James Bradbury}, \bibinfo{person}{Gregory Chanan}, \bibinfo{person}{Trevor Killeen}, \bibinfo{person}{Zeming Lin}, \bibinfo{person}{Natalia Gimelshein}, \bibinfo{person}{Luca Antiga}, \bibinfo{person}{Alban Desmaison}, \bibinfo{person}{Andreas K\"{o}pf}, \bibinfo{person}{Edward Yang}, \bibinfo{person}{Zach DeVito}, \bibinfo{person}{Martin Raison}, \bibinfo{person}{Alykhan Tejani}, \bibinfo{person}{Sasank Chilamkurthy}, \bibinfo{person}{Benoit Steiner}, \bibinfo{person}{Lu Fang}, \bibinfo{person}{Junjie Bai}, {and} \bibinfo{person}{Soumith Chintala}.} \bibinfo{year}{2019}\natexlab{}.
\newblock \bibinfo{booktitle}{\emph{PyTorch: an imperative style, high-performance deep learning library}}.
\newblock


\bibitem[Pearl(2009)]%
        {scm}
\bibfield{author}{\bibinfo{person}{Judea Pearl}.} \bibinfo{year}{2009}\natexlab{}.
\newblock \showarticletitle{{Causal inference in statistics: An overview}}.
\newblock \bibinfo{journal}{\emph{Statistics Surveys}} \bibinfo{volume}{3}, \bibinfo{number}{none} (\bibinfo{year}{2009}), \bibinfo{pages}{96 -- 146}.
\newblock


\bibitem[Penha and Hauff(2023)]%
        {ir3}
\bibfield{author}{\bibinfo{person}{Gustavo Penha} {and} \bibinfo{person}{Claudia Hauff}.} \bibinfo{year}{2023}\natexlab{}.
\newblock \showarticletitle{Do the Findings of Document and Passage Retrieval Generalize to the Retrieval of Responses for Dialogues?} \emph{(\bibinfo{series}{ECIR '23})}. \bibinfo{pages}{132–147}.
\newblock


\bibitem[Rajapakse et~al\mbox{.}(2024)]%
        {ir2}
\bibfield{author}{\bibinfo{person}{Thilina~Chaturanga Rajapakse}, \bibinfo{person}{Andrew Yates}, {and} \bibinfo{person}{Maarten de Rijke}.} \bibinfo{year}{2024}\natexlab{}.
\newblock \showarticletitle{Negative Sampling Techniques for Dense Passage Retrieval in a Multilingual Setting} \emph{(\bibinfo{series}{SIGIR '24})}. \bibinfo{pages}{575–584}.
\newblock


\bibitem[Reimers and Gurevych(2019)]%
        {sbert}
\bibfield{author}{\bibinfo{person}{Nils Reimers} {and} \bibinfo{person}{Iryna Gurevych}.} \bibinfo{year}{2019}\natexlab{}.
\newblock \showarticletitle{Sentence-{BERT}: Sentence Embeddings using {S}iamese {BERT}-Networks} \emph{(\bibinfo{series}{EMNLP '19})}. \bibinfo{pages}{3982–3992}.
\newblock


\bibitem[Rendle and Freudenthaler(2014)]%
        {aobpr}
\bibfield{author}{\bibinfo{person}{Steffen Rendle} {and} \bibinfo{person}{Christoph Freudenthaler}.} \bibinfo{year}{2014}\natexlab{}.
\newblock \showarticletitle{Improving pairwise learning for item recommendation from implicit feedback} \emph{(\bibinfo{series}{WSDM '14})}. \bibinfo{pages}{273–282}.
\newblock


\bibitem[Rendle et~al\mbox{.}(2009)]%
        {bpr}
\bibfield{author}{\bibinfo{person}{Steffen Rendle}, \bibinfo{person}{Christoph Freudenthaler}, \bibinfo{person}{Zeno Gantner}, {and} \bibinfo{person}{Lars Schmidt-Thieme}.} \bibinfo{year}{2009}\natexlab{}.
\newblock \showarticletitle{BPR: Bayesian personalized ranking from implicit feedback} \emph{(\bibinfo{series}{UAI '09})}. \bibinfo{pages}{452–461}.
\newblock


\bibitem[Shi et~al\mbox{.}(2023)]%
        {dnsmn}
\bibfield{author}{\bibinfo{person}{Wentao Shi}, \bibinfo{person}{Jiawei Chen}, \bibinfo{person}{Fuli Feng}, \bibinfo{person}{Jizhi Zhang}, \bibinfo{person}{Junkang Wu}, \bibinfo{person}{Chongming Gao}, {and} \bibinfo{person}{Xiangnan He}.} \bibinfo{year}{2023}\natexlab{}.
\newblock \showarticletitle{On the Theories Behind Hard Negative Sampling for Recommendation} \emph{(\bibinfo{series}{WWW '23})}. \bibinfo{pages}{812–822}.
\newblock


\bibitem[Si et~al\mbox{.}(2022)]%
        {explain1}
\bibfield{author}{\bibinfo{person}{Zihua Si}, \bibinfo{person}{Xueran Han}, \bibinfo{person}{Xiao Zhang}, \bibinfo{person}{Jun Xu}, \bibinfo{person}{Yue Yin}, \bibinfo{person}{Yang Song}, {and} \bibinfo{person}{Ji-Rong Wen}.} \bibinfo{year}{2022}\natexlab{}.
\newblock \showarticletitle{A Model-Agnostic Causal Learning Framework for Recommendation using Search Data} \emph{(\bibinfo{series}{WWW '22})}. \bibinfo{pages}{224--233}.
\newblock


\bibitem[Tao et~al\mbox{.}(2023)]%
        {slmrec}
\bibfield{author}{\bibinfo{person}{Zhulin Tao}, \bibinfo{person}{Xiaohao Liu}, \bibinfo{person}{Yewei Xia}, \bibinfo{person}{Xiang Wang}, \bibinfo{person}{Lifang Yang}, \bibinfo{person}{Xianglin Huang}, {and} \bibinfo{person}{Tat-Seng Chua}.} \bibinfo{year}{2023}\natexlab{}.
\newblock \showarticletitle{Self-Supervised Learning for Multimedia Recommendation}.
\newblock \bibinfo{journal}{\emph{IEEE Transactions on Multimedia}}  \bibinfo{volume}{25} (\bibinfo{year}{2023}), \bibinfo{pages}{5107--5116}.
\newblock


\bibitem[Wang and Liu(2021)]%
        {contrastive}
\bibfield{author}{\bibinfo{person}{Feng Wang} {and} \bibinfo{person}{Huaping Liu}.} \bibinfo{year}{2021}\natexlab{}.
\newblock \showarticletitle{Understanding the Behaviour of Contrastive Loss} \emph{(\bibinfo{series}{CVPR '21})}. \bibinfo{pages}{2495--2504}.
\newblock


\bibitem[Wang et~al\mbox{.}(2017)]%
        {irgan}
\bibfield{author}{\bibinfo{person}{Jun Wang}, \bibinfo{person}{Lantao Yu}, \bibinfo{person}{Weinan Zhang}, \bibinfo{person}{Yu Gong}, \bibinfo{person}{Yinghui Xu}, \bibinfo{person}{Benyou Wang}, \bibinfo{person}{Peng Zhang}, {and} \bibinfo{person}{Dell Zhang}.} \bibinfo{year}{2017}\natexlab{}.
\newblock \showarticletitle{IRGAN: A Minimax Game for Unifying Generative and Discriminative Information Retrieval Models} \emph{(\bibinfo{series}{SIGIR '17})}. \bibinfo{pages}{515–524}.
\newblock


\bibitem[Wang et~al\mbox{.}(2018)]%
        {gan1}
\bibfield{author}{\bibinfo{person}{Qinyong Wang}, \bibinfo{person}{Hongzhi Yin}, \bibinfo{person}{Zhiting Hu}, \bibinfo{person}{Defu Lian}, \bibinfo{person}{Hao Wang}, {and} \bibinfo{person}{Zi Huang}.} \bibinfo{year}{2018}\natexlab{}.
\newblock \showarticletitle{Neural Memory Streaming Recommender Networks with Adversarial Training} \emph{(\bibinfo{series}{KDD '18})}. \bibinfo{pages}{2467–2475}.
\newblock


\bibitem[Wang et~al\mbox{.}(2023)]%
        {nlp2}
\bibfield{author}{\bibinfo{person}{Tianqi Wang}, \bibinfo{person}{Lei Chen}, \bibinfo{person}{Xiaodan Zhu}, \bibinfo{person}{Younghun Lee}, {and} \bibinfo{person}{Jing Gao}.} \bibinfo{year}{2023}\natexlab{}.
\newblock \showarticletitle{Weighted Contrastive Learning With False Negative Control to Help Long-tailed Product Classification} \emph{(\bibinfo{series}{ACL '23})}. \bibinfo{pages}{6930--6941}.
\newblock


\bibitem[Wang et~al\mbox{.}(2021)]%
        {causal3}
\bibfield{author}{\bibinfo{person}{Wenjie Wang}, \bibinfo{person}{Fuli Feng}, \bibinfo{person}{Xiangnan He}, \bibinfo{person}{Xiang Wang}, {and} \bibinfo{person}{Tat-Seng Chua}.} \bibinfo{year}{2021}\natexlab{}.
\newblock \showarticletitle{Deconfounded Recommendation for Alleviating Bias Amplification} \emph{(\bibinfo{series}{KDD '21})}. \bibinfo{pages}{1717–1725}.
\newblock


\bibitem[Wang et~al\mbox{.}(2020)]%
        {rl1}
\bibfield{author}{\bibinfo{person}{Xiang Wang}, \bibinfo{person}{Yaokun Xu}, \bibinfo{person}{Xiangnan He}, \bibinfo{person}{Yixin Cao}, \bibinfo{person}{Meng Wang}, {and} \bibinfo{person}{Tat-Seng Chua}.} \bibinfo{year}{2020}\natexlab{}.
\newblock \showarticletitle{Reinforced Negative Sampling over Knowledge Graph for Recommendation} \emph{(\bibinfo{series}{WWW '20})}.
\newblock


\bibitem[Wei et~al\mbox{.}(2021a)]%
        {causal2}
\bibfield{author}{\bibinfo{person}{Tianxin Wei}, \bibinfo{person}{Fuli Feng}, \bibinfo{person}{Jiawei Chen}, \bibinfo{person}{Ziwei Wu}, \bibinfo{person}{Jinfeng Yi}, {and} \bibinfo{person}{Xiangnan He}.} \bibinfo{year}{2021}\natexlab{a}.
\newblock \showarticletitle{Model-Agnostic Counterfactual Reasoning for Eliminating Popularity Bias in Recommender System} \emph{(\bibinfo{series}{KDD '21})}. \bibinfo{pages}{1791–1800}.
\newblock


\bibitem[Wei et~al\mbox{.}(2021b)]%
        {macr}
\bibfield{author}{\bibinfo{person}{Tianxin Wei}, \bibinfo{person}{Fuli Feng}, \bibinfo{person}{Jiawei Chen}, \bibinfo{person}{Ziwei Wu}, \bibinfo{person}{Jinfeng Yi}, {and} \bibinfo{person}{Xiangnan He}.} \bibinfo{year}{2021}\natexlab{b}.
\newblock \showarticletitle{Model-Agnostic Counterfactual Reasoning for Eliminating Popularity Bias in Recommender System} \emph{(\bibinfo{series}{KDD'21})}. \bibinfo{pages}{1791–1800}.
\newblock


\bibitem[Wei et~al\mbox{.}(2020)]%
        {gnn3}
\bibfield{author}{\bibinfo{person}{Yinwei Wei}, \bibinfo{person}{Xiang Wang}, \bibinfo{person}{Liqiang Nie}, \bibinfo{person}{Xiangnan He}, {and} \bibinfo{person}{Tat-Seng Chua}.} \bibinfo{year}{2020}\natexlab{}.
\newblock \showarticletitle{Graph-Refined Convolutional Network for Multimedia Recommendation with Implicit Feedback} \emph{(\bibinfo{series}{MM '20})}. \bibinfo{pages}{3541–3549}.
\newblock


\bibitem[Wei et~al\mbox{.}(2019)]%
        {gnn4}
\bibfield{author}{\bibinfo{person}{Yinwei Wei}, \bibinfo{person}{Xiang Wang}, \bibinfo{person}{Liqiang Nie}, \bibinfo{person}{Xiangnan He}, \bibinfo{person}{Richang Hong}, {and} \bibinfo{person}{Tat-Seng Chua}.} \bibinfo{year}{2019}\natexlab{}.
\newblock \showarticletitle{MMGCN: Multi-modal Graph Convolution Network for Personalized Recommendation of Micro-video} \emph{(\bibinfo{series}{MM '19})}. \bibinfo{pages}{1437–1445}.
\newblock


\bibitem[Wu et~al\mbox{.}(2021)]%
        {neg_proof2}
\bibfield{author}{\bibinfo{person}{Mike Wu}, \bibinfo{person}{Milan Mosse}, \bibinfo{person}{Chengxu Zhuang}, \bibinfo{person}{Daniel Yamins}, {and} \bibinfo{person}{Noah Goodman}.} \bibinfo{year}{2021}\natexlab{}.
\newblock \showarticletitle{Conditional Negative Sampling for Contrastive Learning of Visual Representations} \emph{(\bibinfo{series}{ICLR'21})}.
\newblock


\bibitem[Wu et~al\mbox{.}(2024)]%
        {spur1}
\bibfield{author}{\bibinfo{person}{Songli Wu}, \bibinfo{person}{Liang Du}, \bibinfo{person}{Jia-Qi Yang}, \bibinfo{person}{Yuai Wang}, \bibinfo{person}{De-Chuan Zhan}, \bibinfo{person}{SHUANG ZHAO}, {and} \bibinfo{person}{Zixun Sun}.} \bibinfo{year}{2024}\natexlab{}.
\newblock \showarticletitle{{RE}-{SORT}: Removing Spurious Correlation in Multilevel Interaction for {CTR} Prediction}. In \bibinfo{booktitle}{\emph{The 40th Conference on Uncertainty in Artificial Intelligence}} \emph{(\bibinfo{series}{UAI'24})}. Article \bibinfo{articleno}{178}, \bibinfo{numpages}{13}~pages.
\newblock


\bibitem[Xu et~al\mbox{.}(2024)]%
        {causalrs}
\bibfield{author}{\bibinfo{person}{Shuyuan Xu}, \bibinfo{person}{Da Xu}, \bibinfo{person}{Evren Korpeoglu}, \bibinfo{person}{Sushant Kumar}, \bibinfo{person}{Stephen Guo}, \bibinfo{person}{Kannan Achan}, {and} \bibinfo{person}{Yongfeng Zhang}.} \bibinfo{year}{2024}\natexlab{}.
\newblock \showarticletitle{Causal Structure Learning for Recommender System}.
\newblock \bibinfo{journal}{\emph{ACM Trans. Recomm. Syst.}} \bibinfo{volume}{3}, \bibinfo{number}{1}, Article \bibinfo{articleno}{8} (\bibinfo{date}{Oct.} \bibinfo{year}{2024}), \bibinfo{numpages}{23}~pages.
\newblock


\bibitem[Xun et~al\mbox{.}(2021)]%
        {mcf2}
\bibfield{author}{\bibinfo{person}{Jiahao Xun}, \bibinfo{person}{Shengyu Zhang}, \bibinfo{person}{Zhou Zhao}, \bibinfo{person}{Jieming Zhu}, \bibinfo{person}{Qi Zhang}, \bibinfo{person}{Jingjie Li}, \bibinfo{person}{Xiuqiang He}, \bibinfo{person}{Xiaofei He}, \bibinfo{person}{Tat-Seng Chua}, {and} \bibinfo{person}{Fei Wu}.} \bibinfo{year}{2021}\natexlab{}.
\newblock \showarticletitle{Why Do We Click: Visual Impression-aware News Recommendation} \emph{(\bibinfo{series}{MM '21})}. \bibinfo{pages}{3881–3890}.
\newblock


\bibitem[Yang et~al\mbox{.}(2024)]%
        {survey}
\bibfield{author}{\bibinfo{person}{Zhen Yang}, \bibinfo{person}{Ming Ding}, \bibinfo{person}{Tinglin Huang}, \bibinfo{person}{Yukuo Cen}, \bibinfo{person}{Junshuai Song}, \bibinfo{person}{Bin Xu}, \bibinfo{person}{Yuxiao Dong}, {and} \bibinfo{person}{Jie Tang}.} \bibinfo{year}{2024}\natexlab{}.
\newblock \showarticletitle{Does Negative Sampling Matter? a Review With Insights Into its Theory and Applications}.
\newblock \bibinfo{journal}{\emph{IEEE Transactions on Pattern Analysis and Machine Intelligence}} \bibinfo{volume}{46}, \bibinfo{number}{8} (\bibinfo{year}{2024}), \bibinfo{pages}{5692--5711}.
\newblock


\bibitem[Yao et~al\mbox{.}(2024)]%
        {minicpmv}
\bibfield{author}{\bibinfo{person}{Yuan Yao}, \bibinfo{person}{Tianyu Yu}, \bibinfo{person}{Ao Zhang}, {and} \bibinfo{person}{et al.}} \bibinfo{year}{2024}\natexlab{}.
\newblock \bibinfo{title}{MiniCPM-V: A GPT-4V Level MLLM on Your Phone}.
\newblock
\showeprint[arxiv]{2408.01800}~[cs.CV]


\bibitem[Ying et~al\mbox{.}(2018)]%
        {pinsage}
\bibfield{author}{\bibinfo{person}{Rex Ying}, \bibinfo{person}{Ruining He}, \bibinfo{person}{Kaifeng Chen}, \bibinfo{person}{Pong Eksombatchai}, \bibinfo{person}{William~L. Hamilton}, {and} \bibinfo{person}{Jure Leskovec}.} \bibinfo{year}{2018}\natexlab{}.
\newblock \showarticletitle{Graph Convolutional Neural Networks for Web-Scale Recommender Systems} \emph{(\bibinfo{series}{KDD '18})}. \bibinfo{pages}{974–983}.
\newblock


\bibitem[Yu et~al\mbox{.}(2024)]%
        {recsys}
\bibfield{author}{\bibinfo{person}{Junliang Yu}, \bibinfo{person}{Hongzhi Yin}, \bibinfo{person}{Xin Xia}, \bibinfo{person}{Tong Chen}, \bibinfo{person}{Jundong Li}, {and} \bibinfo{person}{Zi Huang}.} \bibinfo{year}{2024}\natexlab{}.
\newblock \showarticletitle{Self-Supervised Learning for Recommender Systems: A Survey}.
\newblock \bibinfo{journal}{\emph{IEEE Transactions on Knowledge and Data Engineering}} \bibinfo{volume}{36}, \bibinfo{number}{1} (\bibinfo{year}{2024}), \bibinfo{pages}{335--355}.
\newblock


\bibitem[Zhang et~al\mbox{.}(2023)]%
        {cv1}
\bibfield{author}{\bibinfo{person}{Le Zhang}, \bibinfo{person}{Rabiul Awal}, {and} \bibinfo{person}{Aishwarya Agrawal}.} \bibinfo{year}{2023}\natexlab{}.
\newblock \showarticletitle{Contrasting Intra-Modal and Ranking Cross-Modal Hard Negatives to Enhance Visio-Linguistic Fine-grained Understanding}.
\newblock \bibinfo{journal}{\emph{arXiv preprint arXiv:2306.08832}} (\bibinfo{year}{2023}).
\newblock


\bibitem[Zhang et~al\mbox{.}(2020)]%
        {cv3}
\bibfield{author}{\bibinfo{person}{Shifeng Zhang}, \bibinfo{person}{Cheng Chi}, \bibinfo{person}{Yongqiang Yao}, \bibinfo{person}{Zhen Lei}, {and} \bibinfo{person}{Stan~Z. Li}.} \bibinfo{year}{2020}\natexlab{}.
\newblock \showarticletitle{Bridging the Gap Between Anchor-Based and Anchor-Free Detection via Adaptive Training Sample Selection} \emph{(\bibinfo{series}{CVPR '20})}. \bibinfo{pages}{9756--9765}.
\newblock


\bibitem[Zhang et~al\mbox{.}(2013)]%
        {dns}
\bibfield{author}{\bibinfo{person}{Weinan Zhang}, \bibinfo{person}{Tianqi Chen}, \bibinfo{person}{Jun Wang}, {and} \bibinfo{person}{Yong Yu}.} \bibinfo{year}{2013}\natexlab{}.
\newblock \showarticletitle{Optimizing top-n collaborative filtering via dynamic negative item sampling} \emph{(\bibinfo{series}{SIGIR '13})}. \bibinfo{pages}{785–788}.
\newblock


\bibitem[Zhang et~al\mbox{.}(2022)]%
        {nlp3}
\bibfield{author}{\bibinfo{person}{Yanzhao Zhang}, \bibinfo{person}{Richong Zhang}, \bibinfo{person}{Samuel Mensah}, \bibinfo{person}{Xudong Liu}, {and} \bibinfo{person}{Yongyi Mao}.} \bibinfo{year}{2022}\natexlab{}.
\newblock \showarticletitle{Unsupervised Sentence Representation via Contrastive Learning with Mixing Negatives}.
\newblock \bibinfo{journal}{\emph{Proceedings of the AAAI Conference on Artificial Intelligence}} \bibinfo{volume}{36}, \bibinfo{number}{10} (\bibinfo{date}{Jun.} \bibinfo{year}{2022}), \bibinfo{pages}{11730--11738}.
\newblock


\bibitem[Zhao et~al\mbox{.}(2014)]%
        {sn2}
\bibfield{author}{\bibinfo{person}{Tong Zhao}, \bibinfo{person}{Julian McAuley}, {and} \bibinfo{person}{Irwin King}.} \bibinfo{year}{2014}\natexlab{}.
\newblock \showarticletitle{Leveraging Social Connections to Improve Personalized Ranking for Collaborative Filtering} \emph{(\bibinfo{series}{CIKM '14})}. \bibinfo{pages}{261–270}.
\newblock


\bibitem[Zheng et~al\mbox{.}(2022)]%
        {explain2}
\bibfield{author}{\bibinfo{person}{Yu Zheng}, \bibinfo{person}{Chen Gao}, \bibinfo{person}{Jianxin Chang}, \bibinfo{person}{Yanan Niu}, \bibinfo{person}{Yang Song}, \bibinfo{person}{Depeng Jin}, {and} \bibinfo{person}{Yong Li}.} \bibinfo{year}{2022}\natexlab{}.
\newblock \showarticletitle{Disentangling Long and Short-Term Interests for Recommendation} \emph{(\bibinfo{series}{WWW '22})}. \bibinfo{pages}{2256–2267}.
\newblock


\bibitem[Zhou et~al\mbox{.}(2023a)]%
        {dragon}
\bibfield{author}{\bibinfo{person}{Hongyu Zhou}, \bibinfo{person}{Xin Zhou}, {and} \bibinfo{person}{Zhiqi Shen}.} \bibinfo{year}{2023}\natexlab{a}.
\newblock \showarticletitle{Enhancing Dyadic Relations with Homogeneous Graphs for Multimodal Recommendation}.
\newblock \bibinfo{journal}{\emph{ArXiv}}  \bibinfo{volume}{abs/2301.12097} (\bibinfo{year}{2023}).
\newblock


\bibitem[Zhou and Shen(2023)]%
        {freedom}
\bibfield{author}{\bibinfo{person}{Xin Zhou} {and} \bibinfo{person}{Zhiqi Shen}.} \bibinfo{year}{2023}\natexlab{}.
\newblock \showarticletitle{A Tale of Two Graphs: Freezing and Denoising Graph Structures for Multimodal Recommendation} \emph{(\bibinfo{series}{MM '23})}. \bibinfo{pages}{935–943}.
\newblock


\bibitem[Zhou et~al\mbox{.}(2023b)]%
        {eval1}
\bibfield{author}{\bibinfo{person}{Xin Zhou}, \bibinfo{person}{Hongyu Zhou}, \bibinfo{person}{Yong Liu}, \bibinfo{person}{Zhiwei Zeng}, \bibinfo{person}{Chunyan Miao}, \bibinfo{person}{Pengwei Wang}, \bibinfo{person}{Yuan You}, {and} \bibinfo{person}{Feijun Jiang}.} \bibinfo{year}{2023}\natexlab{b}.
\newblock \showarticletitle{Bootstrap Latent Representations for Multi-modal Recommendation} \emph{(\bibinfo{series}{WWW '23})}. \bibinfo{pages}{845–854}.
\newblock


\bibitem[Zhou et~al\mbox{.}(2023c)]%
        {bm3}
\bibfield{author}{\bibinfo{person}{Xin Zhou}, \bibinfo{person}{Hongyu Zhou}, \bibinfo{person}{Yong Liu}, \bibinfo{person}{Zhiwei Zeng}, \bibinfo{person}{Chunyan Miao}, \bibinfo{person}{Pengwei Wang}, \bibinfo{person}{Yuan You}, {and} \bibinfo{person}{Feijun Jiang}.} \bibinfo{year}{2023}\natexlab{c}.
\newblock \showarticletitle{Bootstrap Latent Representations for Multi-Modal Recommendation} \emph{(\bibinfo{series}{WWW '23})}. \bibinfo{pages}{845–854}.
\newblock


\bibitem[Zhu et~al\mbox{.}(2025)]%
        {mllm}
\bibfield{author}{\bibinfo{person}{Lanyun Zhu}, \bibinfo{person}{Deyi Ji}, \bibinfo{person}{Tianrun Chen}, \bibinfo{person}{Haiyang Wu}, \bibinfo{person}{De~Wen Soh}, {and} \bibinfo{person}{Jun Liu}.} \bibinfo{year}{2025}\natexlab{}.
\newblock \showarticletitle{{CPCF}: A Cross-Prompt Contrastive Framework for Referring Multimodal Large Language Models}. In \bibinfo{booktitle}{\emph{Forty-second International Conference on Machine Learning}}.
\newblock
\urldef\tempurl%
\url{https://openreview.net/forum?id=0MpGi6IwZr}
\showURL{%
\tempurl}


\bibitem[Zhu et~al\mbox{.}(2021)]%
        {pop1}
\bibfield{author}{\bibinfo{person}{Ziwei Zhu}, \bibinfo{person}{Yun He}, \bibinfo{person}{Xing Zhao}, {and} \bibinfo{person}{James Caverlee}.} \bibinfo{year}{2021}\natexlab{}.
\newblock \showarticletitle{Popularity Bias in Dynamic Recommendation} \emph{(\bibinfo{series}{KDD '21})}. \bibinfo{pages}{2439–2449}.
\newblock


\bibitem[Zong et~al\mbox{.}(2024)]%
        {imbalance}
\bibfield{author}{\bibinfo{person}{Daoming Zong}, \bibinfo{person}{Chaoyue Ding}, \bibinfo{person}{Baoxiang Li}, \bibinfo{person}{Jiakui Li}, {and} \bibinfo{person}{Ken Zheng}.} \bibinfo{year}{2024}\natexlab{}.
\newblock \showarticletitle{Balancing Multimodal Learning via Online Logit Modulation}. \bibinfo{pages}{5753--5761}.
\newblock


\end{thebibliography}
